\documentclass{article}

\usepackage[american]{babel}
\usepackage[utf8]{inputenc}
\usepackage[T1]{fontenc}
\usepackage[english=american]{csquotes}
\usepackage{verbatim}
\usepackage{lmodern}

\usepackage{graphics}
\usepackage{amsmath,amssymb,amsthm}
\usepackage{xcolor}

\usepackage{url}
\usepackage{mathtools}
\usepackage{marvosym} 
\usepackage{dsfont} 
\usepackage{enumerate}
\usepackage{float}
\usepackage[normalem]{ulem}
\usepackage{algpseudocode} 
\usepackage{algorithm}
\usepackage[hidelinks]{hyperref}

\usepackage{fullpage}
\usepackage[font={small},format=plain,labelfont=bf]{caption}
\usepackage{authblk}
\usepackage[numbers,sort]{natbib}

\newcommand{\N}{\mathbb{N}} 
\newcommand{\R}{\mathbb{R}} 
\newcommand{\BT}{\mathcal{BT}^{\ast}_n}
\newcommand{\T}{\mathcal{T}^{\ast}_n}
\newcommand{\id}{\textup{id}}

\theoremstyle{definition}
\newtheorem{Def}{Definition}[section]
\newtheorem{Theo}[Def]{Theorem}
\newtheorem{Prop}[Def]{Proposition}
\newtheorem{Lem}[Def]{Lemma}
\newtheorem{Cor}[Def]{Corollary}
\newtheorem{Rem}[Def]{Remark}
\newtheorem{Obs}[Def]{Observation}
\newtheorem{Ex}[Def]{Example}

\newcommand\blfootnote[1]{%
  \begingroup
  \renewcommand\thefootnote{}\footnote{#1}%
  \addtocounter{footnote}{-1}%
  \endgroup
}

\begin{document}

\title{A height-based metaconcept for rooted tree balance and its implications for the $B_1$ index} 

\author[1,$\ast$]{Mareike Fischer}
\author[1]{Tom Niklas Hamann}
\author[2,3]{Kristina Wicke}

\affil[1]{Institute of Mathematics and Computer Science, University of Greifswald, Greifswald, Germany}
\affil[2]{Department of Mathematical Sciences, New Jersey Institute of Technology, Newark, NJ, USA}
\affil[3]{National Institute for Theory and Mathematics in Biology, Northwestern University and The University of Chicago, Chicago, IL, USA}

\date{}
\maketitle

\begin{abstract}
Tree balance has received considerable attention in recent years, both in phylogenetics and in other areas. Numerous (im)balance indices have been proposed to quantify the (im)balance of rooted trees. A recent comprehensive survey summarized this literature and showed that many existing indices are based on similar underlying principles. To unify these approaches, three general metaconcepts were introduced, providing a framework to classify, analyze, and extend imbalance indices. In this context, a metaconcept is a function $\Phi_f$ that depends on another function $f$ capturing some aspect of tree shape. In this manuscript, we extend this line of research by introducing a new metaconcept based on the heights of the pending subtrees of all inner vertices. We provide a thorough analysis of this metaconcept and use it to answer open questions concerning the well-known $B_1$ balance index. In particular, we characterize the tree shapes that maximize the $B_1$ index in two cases: (i) arbitrary rooted trees and (ii) binary rooted trees. For both cases, we also determine the corresponding maximum values of the index.

Finally, while the $B_1$ index is induced by a so-called third-order metaconcept, we explicitly introduce three new (im)balance indices derived from the first- and second-order height metaconcepts, respectively, thereby demonstrating that pending subtree heights give rise to a variety of novel (im)balance indices.
\end{abstract}


\blfootnote{$^\ast$Corresponding author\\ \textit{Email address:} \url{mareike.fischer@uni-greifswald.de, email@mareikefischer.de}}

\section{Introduction}
Tree balance has gained considerable interest in recent years, as it plays a central role in several research areas, including phylogenetics and computer science. In phylogenetics, for example, tree balance is used to investigate differences in diversification rates among groups of organisms (see, e.g., \cite{Kubo1995,Stich2009,Mooers1997}), whereas in computer science it is, for instance, relevant in the context of search trees (see, e.g., \cite{Nievergelt1973,Andersson1993,Roura2013}). Numerous so-called (im)balance indices have been proposed to quantify the (im)balance of rooted trees. Recently, \citet{Fischer2023} provided a comprehensive overview of the current state of research on rooted (im)balance indices. Their survey showed that many indices rely on similar concepts, such as the clade sizes of inner vertices or the depths of leaves. Building on this observation, \citet{Fischer2025} introduced three \emph{metaconcepts} based on values derived from subsets of vertices, namely clade sizes, leaf depths, and balance values, to classify and analyze existing (im)balance indices within a unified framework.

In the balance context, a metaconcept is a function $\Phi_f$ depending on another function $f$, where $f$ measures some aspect of tree shape, and  $\Phi_f$ in turn measures tree (im)balance. In this manuscript, we introduce a new metaconcept based on the heights of the pending subtrees of all inner vertices of a rooted tree. We show that this metaconcept defines an imbalance index for rooted binary trees if the underlying function $f$ is strictly increasing, and that it remains an imbalance index for arbitrary trees if $f$ is additionally $1$-positive, i.e., $f(x) > 0$ for all $x \geq 1$.

We further analyze the extremal trees and corresponding extremal values of this new metaconcept and use these results to answer six open questions posed in \citet{Fischer2023} concerning the so-called $B_1$ index. Introduced by \citet{Shao1990}, the $B_1$ index measures tree balance via the heights of all pending subtrees induced by the inner vertices (except for the root). We determine the maximizing trees of the $B_1$ index in two settings: (i) among arbitrary rooted trees and (ii) among rooted binary trees for all leaf numbers. As we shall see, these sets of trees coincide, meaning that the maximizing trees of the $B_1$ index across arbitrary rooted trees are all binary. We also compute the maximum values of the $B_1$ index in both cases, thereby resolving two additional open questions of \citet{Fischer2023}. Moreover, we answer two open problems regarding the number of trees maximizing the $B_1$ index among all rooted (binary) trees.
These results rely on the relationship between the newly introduced height metaconcept and the $B_1$ index, which turns out to be induced by a so-called third-order height metaconcept. To complement known indices, we additionally introduce one imbalance index induced by a second-order and one balance and one imbalance index, respectively,  induced by a first-order height metaconcept.

The manuscript is organized as follows: In Section \ref{Sec:Prelim}, we present the definitions, notation, and relevant background results needed for our proofs. Section \ref{Sec:Results}, containing our main contributions, is divided into four subsections: In Section \ref{Subsec:Sequences}, we compare four sequences associated with the vertices of a tree; in Section \ref{Subsec:HM}, we analyze the height metaconcept in detail; in Section \ref{Subsec:B1}, we study the $B_1$ index and answer the open problems mentioned above; and in Section \ref{Subsec:new_imb_ind}, we introduce three new (im)balance indices. Finally, in Section \ref{Sec:Discussion}, we conclude with a discussion and an outlook on future research directions.

\section{Preliminaries}
\label{Sec:Prelim}

In this section, we introduce the concepts relevant to the present manuscript, taken mostly from \cite{Fischer2025}, and following the notation of \cite{Fischer2023, Fischer2025}. We begin with some general definitions.

\subsection{Definitions and notation}
\paragraph*{Rooted trees}
A \emph{rooted tree} (or simply \emph{tree}) is a directed graph $T = (V(T),E(T))$, with vertex set $V(T)$ and edge set $E(T)$, containing precisely one vertex of in-degree zero, called the \emph{root} (denoted by $\rho$), such that for every $v \in V(T)$ there exists a unique path from $\rho$ to $v$ and such that there are no vertices with out-degree one. We use $V_L(T) \subseteq V(T)$ to refer to the leaf set of $T$ (i.e., $V_L(T) = \{v \in V(T): \text{out-degree}(v) = 0\}$), and we use $\mathring{V}(T)$ to denote the set of inner vertices of $T$ (i.e., $\mathring{V}(T) = V(T) \setminus V_L(T)$). Moreover, we use $n$ to denote the number of leaves of $T$, i.e., $n = \vert V_L(T) \vert$, which we will also refer to as the \emph{size} of $T$. Note that $\rho \in \mathring{V}(T)$ if $n \geq 2$. If $n=1$, $T$ consists of only one vertex, which is at the same time the root and the tree's only leaf.

A rooted tree is called \emph{binary} if all inner vertices have out-degree two, and for every $n \in \mathbb{N}_{\geq 1}$, we denote by $\BT$ the set of (isomorphism classes of) rooted binary trees with $n$ leaves, and by $\T$ the set of (isomorphism classes of) rooted trees with $n$ leaves. We often call a tree $T \in \T$ an \textit{arbitrary tree}, but remark that arbitrary trees are also sometimes referred to as non-binary trees in the literature (even though binary trees are also contained in the set of arbitrary trees).

\paragraph*{Ancestors, descendants, and cherries}
Let $u,v \in V(T)$ be vertices of $T$. Whenever there exists a path from $u$ to $v$ in $T$, we say that $u$ is an \emph{ancestor} of $v$ and $v$ is a \emph{descendant} of $u$. Note that this implies that each vertex is an ancestor and a descendant of itself. If $u$ and $v$ are connected by an edge, i.e., if $(u,v) \in E(T)$, we also say that $u$ is the \emph{parent} of $v$ and $v$ is a \emph{child} of $u$. The \textit{lowest common ancestor} $LCA_T(u,v)$ of two vertices $u,v \in V(T)$ is the unique common ancestor of $u$ and $v$ that is a descendant of every other common ancestor of them. Moreover, two leaves $x,y \in V_L(T)$ are said to form a \emph{cherry}, denoted by $[x,y]$, if they have the same parent, which is then also called a \textit{cherry parent}.

\paragraph*{Attaching, deleting, and relocating a cherry}
First, by \emph{attaching a cherry} to a tree $T$ to obtain a tree $T'$, we mean replacing a leaf $x \in V_L(T)$ with a cherry. Notice that $T'$ has one more leaf than $T$. Conversely, \emph{deleting a cherry} means replacing a cherry with a single leaf. Finally, \emph{relocating a cherry} in a tree $T$ is defined as performing both operations in sequence: first deleting the cherry and then (re)attaching it to another leaf of $T$. This may change the shape of $T$, but the number of leaves remains the same.

\paragraph*{(Maximal) pending subtrees, clade size (sequence), and standard decomposition}
Given a tree $T$ and a vertex $v \in V(T)$, we denote by $T_v$ the \emph{pending subtree} of $T$ rooted in $v$. We use $n_T(v)$ (or $n_v$ for brevity) to denote the number of leaves in $T_v$, also called the \emph{clade size} of $v$. Based on this, the \textit{clade size sequence} of a tree $T \in \T$ is the list of clade sizes of all its inner vertices, arranged in ascending order (if $n=1$, this list is empty). We denote this sequence by $\mathcal{N}(T) \coloneqq (n_1, \ldots, n_{|\mathring{V}(T)|})$, where $\mathcal{N}(T)_i$ represents the $i$-th entry of $\mathcal{N}(T)$. The length of the clade size sequence for a tree with $n \geq 2$ leaves can range from $1$ to $n-1$. Specifically, the sequence has length 1 if and only if $T$ is a so-called star tree, and it has length $n-1$ if and only if $T$ is binary. Here, the \emph{star tree}, denoted by $T^{star}_n$, is the rooted tree with $n$ leaves that either satisfies $n=1$, or $n \geq 2$ and has a single inner vertex (the root), which is adjacent to all leaves (see Figure \ref{Fig:special_trees} for an example with $n = 6$ leaves). We will often decompose a rooted tree $T$ on $n \geq 2$ leaves into its maximal pending subtrees rooted in the children of $\rho$. We denote this decomposition as $T = (T_{v_1}, \ldots, T_{v_k})$, where $v_1, \ldots, v_k$ are the children of the root in $T$, and refer to it as the \emph{standard decomposition} of $T$. If $T$ is binary, we have $k=2$, and thus $T = (T_{v_1}, T_{v_2})$.

\paragraph*{Depth, leaf depth sequence, and height (value)}
The \emph{depth} $\delta_T(v)$ (or $\delta_v$ for brevity) of a vertex $v \in V(T)$ is the number of edges on the path from the root $\rho$ to $v$. Based on this, the \textit{leaf depth sequence} of a tree $T \in \T$ is the list of leaf depths of all its leaves, arranged in ascending order. We denote this sequence by $\Delta(T) \coloneqq (\delta_1, \ldots, \delta_{n})$, where $\Delta(T)_i$ represents the $i$-th entry of $\Delta(T)$. Unlike the clade size sequence, the leaf depth sequence always has length $n = |V_L(T)|$, regardless of whether the tree is binary. Moreover, the height $h(T)$ of $T$ is the maximum depth of any leaf, i.e., $h(T) = \max_{x \in V_L(T)} \delta_T(x)$. Further, the \textit{height value} of a vertex $v$ is the height of its pending subtree $T_v$, i.e., $h_T(v) = h(T_v)$ (or $h_v$ for brevity).

\paragraph*{Balance value (sequence) and balanced vertices}
Now let $T$ be a rooted binary tree with at least two leaves and let $v \in \mathring{V}(T)$ be an inner vertex of $T$ with children $v_1$ and $v_2$. The \emph{balance value} $b_T(v)$ (or $b_v$ for brevity) of $v$ is defined as $b_{T}(v) \coloneqq |n_{v_1} - n_{v_2}|$. Based on this, the \emph{balance value sequence} of a binary tree $T \in \BT$ is the list of balance values of all its inner vertices, arranged in ascending order (if $n=1$, this list is empty). We denote this sequence by $\mathcal{B}\left(T\right) \coloneqq (b_1, \ldots, b_{n-1})$. The $i$-th entry of $\mathcal{B}(T)$ is denoted by $\mathcal{B}\left(T\right)_i$. Note that for any $T \in \BT$, the length of $\mathcal{B}(T)$ is $n-1$. Moreover, an inner vertex $v$ is called \emph{balanced} if it fulfills $b_T(v) \leq 1$.

\begin{figure}[ht]
\centering
	\includegraphics[scale=2]{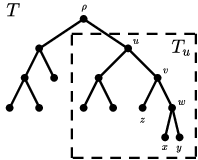}
	\caption{Rooted binary tree $T$ with eight leaves and root $\rho$ (figure adapted from \cite{Fischer2025}). Vertices $\rho$, $u$, and $v$ are ancestors of $v$. The parent of $v$ is $u$, and $v$ is one of two children of $u$. The descendants of $v$ are $v$, $z$, $w$, $x$, and $y$. The lowest common ancestor of $x$ and $z$ is $LCA_T(x,z) = v$. The leaves $x$ and $y$ form the cherry $[x,y]$, whose parent is $w$. The pending subtree of $u$ is $T_u$, which is also one of the two maximal pending subtrees of $T$. It has five leaves and thus $n_T(u) = 5$, i.e., the clade size of $u$ is five. The balance value of $v$ is one, i.e., $b_T(v) = 2-1 = 1$, hence $v$ is balanced. Vertices $u$ and $w$ are balanced, too, because $b_T(u) = 3-2 = 1$ and $b_T(w) = 1-1 = 0$. The root $\rho$ is not balanced as $b_T(\rho) = 5-3 = 2$.}
	\label{Fig:examplesDefinitions}
\end{figure}

\paragraph*{Important (families of) trees}
Next, in addition to the previously introduced star tree, we introduce several specific families of trees that will play an important role throughout this manuscript (see Figure~\ref{Fig:special_trees} for examples).

First, the \emph{fully balanced tree of height $h$} (or fb-tree for brevity), denoted by $T^{fb}_h$, is the rooted binary tree with $n = 2^h$ leaves with $h \in \N_{\geq 0}$ in which all leaves have depth precisely $h$. Note that $T^{fb}_h$ contains $2^{h} = n$ pending subtrees of height $0$ (the leaves), $2^{h-1}$ pending subtrees of height $1$ (the cherries) and, in general, $2^{h-i}$ pending subtrees of height $i$. Further, for $h \geq 1$, we have $T^{fb}_h = \left(T^{fb}_{h-1}, T^{fb}_{h-1}\right)$.

Second, the \emph{maximally balanced tree} (or mb-tree for brevity), denoted by $T^{mb}_n$, is the rooted binary tree with $n$ leaves in which all inner vertices are balanced. Recursively, a rooted binary tree with $n \geq 2$ leaves is maximally balanced if its root is balanced and its two maximal pending subtrees are maximally balanced. Notice that for $h \in \N_{\geq 0}$, we have $T^{mb}_{2^h} = T^{fb}_h$.

Third, the \emph{greedy from the bottom tree} (or gfb-tree for brevity), denoted by $T^{gfb}_n$, is the rooted binary tree with $n$ leaves that results from greedily clustering trees of minimal leaf numbers, starting with $n$ single vertices and proceeding until only one tree is left as described by~\cite[Algorithm 2]{Coronado2020a}. Another way to construct the gfb-tree is a direct consequence of \cite[Lemma 4.17]{Cleary2025}: The gfb-tree $T^{gfb}_n$ with $n$ leaves can be constructed from the fb-tree $T^{fb}_h$ with $h = \lfloor\log_2(n)\rfloor$ by attaching $n - 2^h$ cherries from left to right to the leaves of $T^{fb}_h$. Consequently, for $h \in \N_{\geq 0}$, we have $T^{gfb}_{2^h} = T^{fb}_h$.

\begin{Rem}
\label{Rem:gfb_properties}
    Note that the second definition of the gfb-tree implies that the gfb-tree $T^{gfb}_n$ with $n$ leaves can also be constructed from the fb-tree $T^{fb}_h$ with $h = \lceil\log_2(n)\rceil$ by deleting $2^h - n$ cherries from right to left. In particular, the gfb-tree has minimal height $h\left(T^{gfb}_n\right) = \lceil\log_2(n)\rceil$. Moreover, if $n$ is even, all leaves of the gfb-tree are part of a cherry, whereas if $n$ is odd, all leaves except for one are part of a cherry.

    Additionally, the gfb-tree with an even number of leaves $n$ can be obtained from the gfb-tree with $n-1$ leaves by attaching a cherry to its unique leaf that is not part of a cherry. Further, for even $n$, deleting all $\frac{n}{2}$ cherries of $T^{gfb}_n$ yields $T^{gfb}_{\frac{n}{2}}$. Conversely, attaching a cherry to each leaf of $T^{gfb}_{\frac{n}{2}}$ yields $T^{gfb}_n$.
\end{Rem}

Fourth, the \emph{caterpillar tree} (or simply \emph{caterpillar}), denoted by $T^{cat}_n$, is the rooted binary tree with $n$ leaves defined as follows: if $n=1$, it is a single leaf; if $n \geq 2$, it is a rooted binary tree with $n$ leaves that contains exactly one cherry.

Finally, the \textit{binary echelon tree} (or simply \textit{echelon tree}) from \cite{Currie2024}, denoted by $T^{be}_n$, can be defined as follows \cite{collessMinussackin2025}: Let $k_n = \lceil\log_2(n)\rceil$ and consider the binary expansion of $n$, i.e., $n = \sum\limits_{i=0}^{k_n} \alpha_i\cdot 2^i$, where $\alpha_i$ equals $1$ precisely if $2^i$ is contained in the binary expansion of $n$ and $0$ else. Note that $w(n) = \sum\limits_{i=0}^{k_n} \alpha_i$ is the so-called \textit{binary weight} of $n$. Let $f_n(i)$ for $i = 1, \ldots, w(n)$ denote the $i$-th power of two in the binary expansion of $n$ when they are sorted in ascending order. For instance, if $n = 11 = 2^0+2^1+2^3$, we have $f_{11}(1) = 2^0$, $f_{11}(2) = 2^1$ and $f_{11}(3) = 2^3$. Note that this implies $n = \sum\limits_{i=1}^{w(n)} f_n(i)$. We construct tree $T_n^{be}$ as follows \cite{collessMinussackin2025}: We start with a \enquote{top-caterpillar} $T^{cat}_{w(n)}$ on $w(n)$ leaves and replace its leaves by pending fb-subtrees of height $f_n(i)$. If $w(n) = 1$, i.e., $n = 2^{k_n}$ is a power of two, the caterpillar has only one leaf, which is replaced by $T^{fb}_{k_n}$. If $w(n) \geq 2$, we start by replacing one of the top caterpillar's cherry leaves with $T^{fb}_j$, where $j = \log_2(f_n(1))$ is the smallest index such that $\alpha_j = 1$. Then, we take one unvisited leaf of the original caterpillar at a time, starting with the one furthest from the root and ending with the one adjacent to the root, and replace the respective leaf by $T^{fb}_i$, where $i$ is the smallest index with $\alpha_i = 1$ which has not been considered yet. Note that this construction implies that $T^{fb}_{\log_2(f_n(w(n)))} = T^{fb}_{k_n-1}$ is a maximal pending subtree of the root of $T^{be}_n$. In particular, the standard decomposition of $T^{be}_n$ is given by $\left(T^{fb}_{k_n-1},T^{be}_{n-2^{k_n-1}}\right)$.

Notice that all trees introduced above (including the star tree) are unique (up to isomorphism) and have the property that all their pending subtrees are (smaller) trees of the same type. Moreover, we remark that the caterpillar is generally regarded as the most imbalanced (binary) tree, whereas the fully balanced tree is considered the most balanced binary tree when it exists, i.e., for leaf numbers that are powers of two. For other leaf numbers, both the maximally balanced tree and the greedy from the bottom tree are often regarded as the most balanced binary trees, whereas the star tree is usually considered to be the most balanced arbitrary tree. For more details see \cite{Fischer2023}.

\begin{figure}[ht]
\centering
	\includegraphics[width=0.7\textwidth]{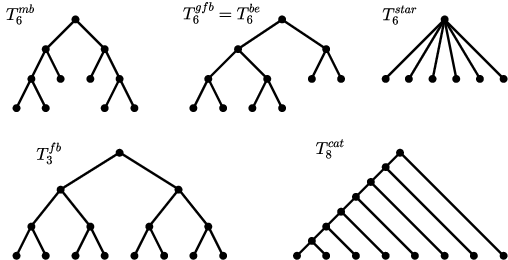}
	\caption{Examples of the special trees considered throughout this manuscript (figure adapted from \cite{Fischer2025}). Note that the gfb-tree and the echelon tree generally do not coincide. For example, for $n = 5$, we have $T^{gfb}_5 = \left(T^{cat}_3,T^{cat}_2\right) \neq \left(T^{fb}_2,T^{cat}_1\right) = T^{be}_5$ (see Figure \ref{Fig:5_2_3} in the appendix).}
	\label{Fig:special_trees}
\end{figure}

\paragraph*{(Im)balance index, locality and recursiveness}
We next introduce the concept of a tree (im)balance index. First, following~\citet{Fischer2023}, a \emph{(binary) tree shape statistic} is a function $t: \T$ (respectively, $\BT$) $\rightarrow \mathbb{R}$ that depends only on the shape of $T$ but not on the labeling of vertices or the lengths  of edges. Based on this, a tree (im)balance index is defined as follows:
\begin{Def}[(Binary) (im)balance index~(\citet{Fischer2023})]
    A (binary) tree shape statistic $t$ is called an \emph{imbalance index} if and only if
	\begin{enumerate}[(i)]
		\item the caterpillar $T^{cat}_n$ is the unique tree maximizing $t$ on its domain $\T$, or $\BT$, respectively, for all $n \geq 1$,
		\item the fully balanced tree $T^{fb}_h$ is the unique tree minimizing $t$ on $\BT$ for all $n = 2^h$ with $h \in \N_{\geq 0}$.
	\end{enumerate}
    If the domain of $t$ is $\BT$, we often call $t$ a \textit{binary imbalance index} to highlight this fact.

    We note that, in addition to imbalance indices, there are also balance indices, such as the $B_1$ index (defined below). A balance index is minimized by the caterpillar and maximized by the fb-tree. Also note that a balance index can be obtained from an imbalance index (and vice versa) simply by multiplying by $-1$.
\end{Def}

Given two trees $T,T' \in \T$ (or $\BT$, respectively) and an imbalance index (or balance index) $t$, we say that \emph{$T$ is more balanced than $T'$} with respect to $t$ if $t(T) < t(T')$ (or $t(T) > t(T')$, respectively). More generally, when we say that a tree $T$ \emph{minimizes an imbalance index}, we mean that it attains the minimum value of the index among all trees in the same domain $\T$ (or $\BT$, respectively). Similarly, when comparing a tree $T$ to a family of trees (such as the ones defined above), we always compare it to the representative of the family with the same number of leaves.

Further, two (im)balance indices $\varphi_1$ and $\varphi_2$ are said to be \emph{equivalent} if, for all trees $T_1, T_2 \in \T$ (or $\BT$, respectively), the following holds: $\varphi_1(T_1) < \varphi_1(T_2) \Longleftrightarrow \varphi_2(T_1) < \varphi_2(T_2)$. In other words, equivalence means that $\varphi_1$ and $\varphi_2$ induce the same ranking of trees from most balanced to least balanced.
\smallskip

We next turn to two desirable properties of (im)balance indices, namely locality and recursiveness.

\begin{Def}[Locality~(\citet{Mir2013,Fischer2023})]
    An (im)balance index $t$ is called \emph{local} if it satisfies $t(T) - t(T') = t(T_v) - t\left(T'_v\right)$ for all $v \in V(T)$ and for all pairs $T,T' \in \T$ (or $\BT$, respectively) such that $T'$ can be obtained from $T$ by replacing the pending subtree $T_v$ rooted in $v$ with a (binary) tree $T'_v$ that has the same number of leaves and is also rooted in $v$.
\end{Def}

In other words, if $t$ is local and two trees $T$ and $T'$ differ only in a pending subtree, then the difference in their $t$-values is equal to the difference in the $t$-values of the corresponding pending subtrees.
\smallskip

Next, we introduce the recursiveness of a tree shape statistic.

\begin{Def}[Recursiveness (based on~\citet{Matsen2007,Fischer2023})]
\label{Def:recursivenessBook}
    A \emph{recursive tree shape statistic} of length $x \in \N_{\geq 1}$ is an ordered pair $(\lambda,r)$, where $\lambda \in \R^{x}$ and $r$ is an $x$-vector of symmetric functions each mapping a multiset of $x$-vectors to $\R$. In this definition, $x$ is the number of recursions that are used to calculate the index, the vector $\lambda$ contains the start value for each of the $x$ recursions, i.e., the values of $T \in \mathcal{T}^{\ast}_1$ if $n = 1$, and the vector $r$ contains the recursions themselves. In particular, $r_i(T) = \lambda_i$ for $n = 1$, and for $T = (T_1, \ldots, T_k)$, recursion $r_i$ operates on $k$ vectors of length $x$, namely $(r_1(T_1), \ldots, r_x(T_1)), \ldots, (r_1(T_k), \ldots, r_x(T_k))$, each representing one of the maximal pending subtrees $T_1, \ldots, T_k$ and containing their respective values. The recursions are symmetrical functions, i.e., the order of those $k$ vectors is permutable, because we are solely considering unordered trees. If only binary trees are considered, i.e., $k = 2$ for every pending subtree, we use the term \emph{binary recursive tree shape statistic}.
\end{Def}

In the following, we recall the general (im)balance index metaconcept, as introduced by \citet{Fischer2025}. 
Based on this framework, we then introduce our new height metaconcept, which is defined in terms of the so-called height sequence. Finally, we introduce the related $B_1$ index, a well-known balance index.

\paragraph*{Height sequence, height metaconcept, and $B_1$ index}Let $T \in \mathcal{T} \subseteq \T$ be a tree, and let $Seq(T)$ be a vertex value sequence on a subset $V' \subseteq V(T)$, i.e., a sequence that assigns each vertex $v \in V'$ a value $s_v$ derived from $v$. Assume that $Seq(T)$ is sorted in ascending order, and let $Seq(T)_i$ denote its $i$-th entry.

Furthermore, let $\omega \in \N_{\geq 1}$, $c = \min\limits_{T' \in \mathcal{T}}\left\{Seq(T')_1\right\}$, i.e., $c$ is the smallest possible sequence value that a tree in $\mathcal{T}$ can attain, and $f: \R_{\geq c} \times \R^{\omega-1} \rightarrow \R$ be a function that depends on an entry of $Seq(T)$ and $\omega-1$ additional values $o_1(T), \ldots, o_{\omega-1}(T)$, such as the number of inner vertices, i.e., $o_i(T) = |\mathring{V}(T)|$, or the number of leaves of $T$, i.e., $o_i(T) = n$. Then,
\begin{align*}
    \Phi^{Seq}_{f}(T) &\coloneqq \sum\limits_{s \in Seq(T)} f(s,o_1(T), \ldots, o_{\omega-1}(T))
\end{align*}
is called the \textit{(im)balance index metaconcept of order $\omega$}. Clearly,
\begin{align*}
    \Phi^{Seq}_{f}(T)
    = \sum\limits_{v \in V'} f(s_v,o_1(T), \ldots, o_{\omega-1}(T))
    = \sum\limits_{i = 1}^{|V'|} f(Seq(T)_i,o_1(T), \ldots, o_{\omega-1}(T)).
\end{align*}

We now specialize the general metaconcept by using a particular sequence. The \textit{height sequence} of a tree $T \in \T$ is defined as the list of the height values of all its inner vertices, arranged in ascending order (if $n=1$, this list is empty). We denote this sequence by $\mathcal{H}(T) \coloneqq (h_1, \ldots, h_{|\mathring{V}(T)|})$, and refer to its $i$-th entry as $\mathcal{H}(T)_i$. The length of the height sequence of a tree with $n \geq 2$ leaves ranges from $1$ to $n-1$. In particular, the sequence has length $1$ if and only if $T$ is a star tree, and length $n-1$ if and only if $T$ is binary.

Based on this sequence, for a rooted tree $T \in \T$, we define the \textit{height metaconcept} (HM) of order $\omega$ as
\[\Phi^{\mathcal{H}}_{f}(T) \coloneqq \sum\limits_{h \in \mathcal{H}(T)} f(h, o_1(T), \ldots, o_{\omega -1}(T)).\]
Since $1$ is the smallest height value attained by an inner vertex, we have $c = 1$, and consequently $f: \R_{\geq 1} \times \R^{\omega-1} \rightarrow \R$. Throughout, if the order of the metaconcept is $\omega = 1$, we frequently consider functions satisfying the property $f(x) > 0$ for all $x \geq 1$ and call them \textit{$1$-positive}.

Note that \citet{Fischer2025} introduced three metaconcepts based on the balance value sequence, the clade size sequence, and the leaf depth sequence, respectively, and analyzed their ability to measure imbalance. For this reason, and given that the majority of known indices are formulated as measures of imbalance, we later analyze in which cases the height metaconcept can also serve as a measure of imbalance. Note again that an imbalance index can be easily transformed into a balance index by replacing $f$ with $-f$.

\medskip

In this work, we answer several open questions from the literature concerning the following well-known balance index. \citet{Shao1990} defined the \textit{$B_1$ index} as
\[B_1(T) \coloneqq \sum\limits_{v \in \mathring{V}(T) \setminus \left\{\rho\right\}} \frac{1}{h(T_v)} = \sum\limits_{v \in \mathring{V}(T) \setminus \left\{\rho\right\}} \frac{1}{h_T(v)}.\]
This balance index is defined for arbitrary trees, but it satisfies only the definition of a binary balance index, since it considers the star tree, whose $B_1$ index is an empty sum and thus equals 0, to be more imbalanced than the caterpillar (\citet{Fischer2023}).

We will see in Remark \ref{Rem:B1_HM} that the $B_1$ index is induced by the third-order HM.

\subsection{Known results}
Before presenting our new findings, we recall two results from the literature. The first concerns a functional that has the same form as the metaconcepts; by choosing $Seq = \mathcal{H}$ in this lemma, we obtain the height metaconcept. The second result pertains to the gfb-tree.

\begin{Lem}[\citet{Fischer2025}, Lemma 3.4]
\label{Lem:Seq_entries_Min_Max_meta_binary}
    Let $Seq$ be a sequence of length $l$, sorted in ascending order, which can be determined for every tree $T \in \mathcal{T}$, where $\mathcal{T} \subseteq \T$. Denote the $i$-th entry of $Seq(T)$ by $Seq(T)_i$. Let $f: \R \rightarrow \R$ be a function, and define the functional $\Phi^{Seq}_f: \mathcal{T} \rightarrow \R$ by
    \[\Phi^{Seq}_f(T) \coloneqq \sum\limits_{i=1}^{l} f\left(Seq(T)_i\right).\]
    Then, we have:
    \begin{enumerate}
        \item
            \begin{enumerate}
                \item If a tree $T \in \mathcal{T}$ minimizes the functional $\Phi^{Seq}_f$ on $\mathcal{T}$ for all strictly increasing functions $f$, then for all $\widetilde{T} \in \mathcal{T}$, we have
                \[Seq(T)_i \leq Seq(\widetilde{T})_i \text{ for all } i \in \left\{1,\ldots,l\right\}.\]
                \item Conversely, if a tree $T \in \mathcal{T}$ satisfies for all $\widetilde{T} \in \mathcal{T}$ and all $i \in \left\{1,\ldots,l\right\}$  
                \[Seq(T)_i \leq Seq(\widetilde{T})_i,\]
                then $T$ minimizes the functional $\Phi^{Seq}_f$ on $\mathcal{T}$ for all (not necessarily strictly) increasing functions.
            \end{enumerate}

        \item
            \begin{enumerate}
                \item If a tree $T \in \mathcal{T}$ \emph{uniquely} minimizes the functional $\Phi^{Seq}_f$ on $\mathcal{T}$ for some increasing function $f$, then for all $\widetilde{T} \in \mathcal{T} \setminus \left\{T\right\}$, we have
                \[Seq(T)_i < Seq(\widetilde{T})_i \text{ for at least one } i \in \left\{1,\ldots,l\right\}.\]
                \item Conversely, if a tree $T \in \mathcal{T}$ satisfies for all $\widetilde{T} \in \mathcal{T} \setminus \left\{T\right\}$
                \[Seq(T)_i \leq Seq(\widetilde{T})_i \text{ for all } i \in \left\{1,\ldots,l\right\}\]
                and
                \[Seq(T)_i < Seq(\widetilde{T})_i \text{ for at least one } i \in \left\{1,\ldots,l\right\},\]
                then $T$ (uniquely) minimizes the functional $\Phi^{Seq}_f$ on $\mathcal{T}$ for all (strictly) increasing functions $f$.
        \end{enumerate}
    \end{enumerate}

    Both statements also hold in the maximization case, where \enquote{minimizing} is replaced by \enquote{maximizing}, and all inequalities are reversed.
\end{Lem}

Next, we turn to the gfb-tree and its pending subtree sizes.

\begin{Prop}[\citet{Coronado2020a}, Proposition 5]
\label{Prop:gfb_Ta_or_Tb_fb}
    Let $T^{gfb}_n = (T_1,T_2)$ be a gfb-tree with $n$ leaves, $T_1 \in \mathcal{BT}^{\ast}_{n_1}$, $T_2 \in \mathcal{BT}^{\ast}_{n_2}$ and $n_1 \geq n_2$. Let $n = 2^{h_n} +p_n$ with $h_n = \lfloor\log_2(n)\rfloor$ and $0 \leq p_n < 2^{h_n}$. Then, we have:
    \begin{enumerate}
        \item If $0 \leq p_n \leq 2^{h_n-1}$, then $n_1 = 2^{h_n-1} +p_n$, $n_2 = 2^{h_n-1}$, and $T_2$ is an fb-tree.
        \item If $2^{h_n-1} \leq p_n < 2^{h_n}$, then $n_1 = 2^{h_n}$, $n_2 = p_n$, and $T_1$ is an fb-tree.
    \end{enumerate}
\end{Prop}

We are now in a position to state our new results.

\section{Results}
\label{Sec:Results}

\paragraph*{Summary of our main results}
We prove that the height metaconcept (HM) induces an imbalance index for all strictly increasing and $1$-positive functions $f$, and a binary imbalance index for all strictly increasing functions $f$. Moreover, we show that the minimizing trees are the same for all strictly increasing functions $f$ on $\BT$; in particular, we fully characterize these minimizing trees by showing that the gfb-tree is always among them and that all binary minimizing trees share the same height sequence as the gfb-tree.

By exploiting the relationship between the HM and the $B_1$ index, we prove that the set of trees minimizing the HM on $\BT$ for all increasing functions $f$ coincides with the set of trees maximizing the $B_1$ index on $\BT$ and $\T$. This yields the same complete characterization as above, i.e., the gfb-tree always maximizes the $B_1$ index and all other maximizing trees have the same height sequence as the gfb-tree. This resolves two open problems posed by \citet{Fischer2023}.

Two further open problems of \citet{Fischer2023} concern the maximum value of the $B_1$ index on $\BT$ and $\T$, for which we present explicit formulas. We also determine the leaf numbers for which the  maximizing tree of the $B_1$ index is unique in $\BT$ and $\T$, respectively. Finally, we show that the HM is not local but recursive for all functions $f$, without imposing any further restrictions. This generalizes an earlier result on the $B_1$ index (\citet[Propositions 10.2 and 10.3]{Fischer2023}) to the HM.
\medskip

Before we analyze the HM and the $B_1$ index, we begin by comparing the height sequence with the balance value sequence, the clade size sequence, and the leaf depth sequence in order to show that all these sequences contain different information on the underlying tree and have different properties. Recall that \citet{Fischer2025} introduced three metaconcepts based on these latter three sequences, whereas our new metaconcept is based on the height sequence.

\subsection{Comparing sequences}
\label{Subsec:Sequences}
We start with an observation concerning the height sequence. As shown in \cite{Fischer2025}, the balance value sequence, the clade size sequence, and the leaf depth sequence of a tree can each be derived from the corresponding sequences of its maximal pending subtrees. This recursive property also holds for the height sequence.

\begin{Obs}
\label{Obs:H_recursiveness}
Let $T \in \T$ be a tree with standard decomposition $T = (T_1, \ldots, T_k)$. The height sequence of $T$ consists of all entries of the height sequences of its maximal pending subtrees, together with the maximum height among these pending subtrees plus one. This holds because every inner vertex of a maximal pending subtree has the same height value in $T$ as in its corresponding pending subtree $T_i$. Moreover, the maximum height of all maximal pending subtrees plus one equals the height of $T$ and thus the height value of its root.
\end{Obs}

\citet{Fischer2025} compared the three sequences they introduced by considering pairs of sequences and listing minimal examples of two binary trees that share both sequences, only one, or none. We complement the comparison by additionally taking the height sequence into account. Specifically, we compare the height sequence with each of the three other sequences. Once again, we find minimal examples, most of them unique, of pairs of binary trees that share both, only one, or none of the two respective sequences under investigation. For examples and a detailed overview, see Figures \ref{Fig:5_2_3} to \ref{Fig:11_194_199} in the appendix and Table \ref{Tab:Uebersicht_Bsp_Seq}. This shows that all these sequences that can be derived from a tree have inherently different properties, which is the basis of our investigation of the height sequence, which -- unlike the other sequences mentioned here -- has not yet been analyzed in the context of tree balance metaconcepts.

\begin{table}
	\centering
	\caption{This table provides an overview of where to find minimal examples of pairs of binary trees that share or do not share the height sequence $\mathcal{H}$ and one of the sequences $\mathcal{B}$, $\mathcal{N}$, or $\Delta$. Note that the figures referred to in the table can be found in the appendix.}
	\label{Tab:Uebersicht_Bsp_Seq}
	\begin{tabular}{ c | c | c | c}
	\textbf{First Sequence} & \textbf{Second Sequence} & \textbf{Figure} & $\boldsymbol{n}$\\
	\hline \hline
    $\mathcal{H}(T_1) = \mathcal{H}(T_2)$ & $\mathcal{B}(T_1) = \mathcal{B}(T_2)$ & \ref{Fig:11_194_199} & $11$\\
	$\mathcal{H}(T_1) = \mathcal{H}(T_2)$ & $\mathcal{B}(T_1) \neq \mathcal{B}(T_2)$ & \ref{Fig:5_2_3} & $5$\\
    $\mathcal{H}(T_1) \neq \mathcal{H}(T_2)$ & $\mathcal{B}(T_1) = \mathcal{B}(T_2)$ & \ref{Fig:9_42_44} & $9$\\
    \hline
    $\mathcal{H}(T_1) = \mathcal{H}(T_2)$ & $\mathcal{N}(T_1) = \mathcal{N}(T_2)$ & \ref{Fig:11_194_199} & $11$\\
    $\mathcal{H}(T_1) = \mathcal{H}(T_2)$ & $\mathcal{N}(T_1) \neq \mathcal{N}(T_2)$ & \ref{Fig:5_2_3} & $5$\\
    $\mathcal{H}(T_1) \neq \mathcal{H}(T_2)$ & $\mathcal{N}(T_1) = \mathcal{N}(T_2)$ & \ref{Fig:9_42_44} & $9$\\
    \hline
	$\mathcal{H}(T_1) = \mathcal{H}(T_2)$ & $\Delta(T_1) = \Delta(T_2)$ & \ref{Fig:8_20_22} & $8$\\
	$\mathcal{H}(T_1) = \mathcal{H}(T_2)$ & $\Delta(T_1) \neq \Delta(T_2)$ & \ref{Fig:5_2_3} & $5$\\
	$\mathcal{H}(T_1) \neq \mathcal{H}(T_2)$ & $\Delta(T_1) = \Delta(T_2)$ & \ref{Fig:6_5_6} & $6$
	\end{tabular}
\end{table}

\subsection{Height metaconcept \texorpdfstring{$\Phi^{\mathcal{H}}_f$}{PhiHf}}
\label{Subsec:HM}
In this section, we first outline the relationship between the $B_1$ index and the HM, followed by an analysis of the trees that maximize and minimize the HM. We then identify the families of functions for which the HM induces a (binary) imbalance index and establish several properties of the trees that minimize the HM on $\BT$. The most important property is that all trees minimizing the HM for all increasing functions $f$ have the same height sequence as the gfb-tree, which completely characterizes this set of trees. After proving that the HM is not local but recursive, we answer several open questions concerning the $B_1$ index by exploiting the previously established results for the HM.

While we focus on the first-order HM and the $B_1$ index, all results concerning minimizing and maximizing trees also apply to functions that are equivalent to the first-order metaconcept. For some illustrative examples, see the next remark.

\begin{Rem}[adapted from \citet{Fischer2025}, Remark 3.2]
\label{Rem:H_foMeta_equiv_higherorderMeta}
    The binary HM of order $\omega \geq 2$ with function $f(x,o_1, \ldots, o_{\omega-1}) = f_1(x) \cdot f_2(o_1, \ldots, o_{\omega-1}) + f_3(o_{1}, \ldots, o_{\omega-1})$ is equivalent to the first-order HM with $f(x) = f_1(x)$, provided that $f_2(o_1, \ldots, o_{\omega-1}) > 0$ and the additional values $o_1, \ldots, o_{\omega-1}$ are the same for every tree with the same number of leaves. Note that as $f_3$ is independent of $x$ and thus can be regarded as a summation constant, it only causes a constant shift in the definition of $f$ which is identical for all trees in $\BT$ and therefore does not have an impact on the rankings induced by $f$. This explains the equivalence of $f$ and $f_1$ under the aforementioned circumstances concerning $f_2$. However, recall that all trees in $\BT$ have the same number of leaves $n$ and the same number of inner vertices $(n-1)$. Therefore, each $o_i$ may depend on the leaf number $n$, but not on tree-specific properties such as the height. By contrast, for arbitrary trees the number of inner vertices, i.e., the number of summands in the HM, may vary, so this equivalence does not generally hold. Consequently, the HM of order $\omega \geq 2$ can only be guaranteed to be equivalent to the first-order metaconcept if the function has the form $f(x,o_1, \ldots, o_{\omega-1}) = f_1(x) \cdot f_2(o_1, \ldots, o_{\omega-1})$, where, again, $f_2(o_1, \ldots, o_{\omega-1}) > 0$ and the additional values $o_1, \ldots, o_{\omega-1}$ are the same for every tree with the same number of leaves. In this case, the metaconcept remains equivalent to the first-order metaconcept with function $f(x) = f_1(x)$.
\end{Rem}

After analyzing the HM, we can derive several properties of the $B_1$ index from it. For a first insight into their relationship, see the following remark.

\begin{Rem}
\label{Rem:B1_HM}
    The $B_1$ index is induced by the third-order HM with $f_{B_1}(h,|\mathring{V}(T)|,h(T)) = \frac{1}{h} - \frac{h(T)}{|\mathring{V}(T)|}$, because for any $T \in \T$, we have
    \[\Phi^{\mathcal{H}}_{f_{B_1}}(T) = \sum\limits_{h \in \mathcal{H}(T)} \left(\frac{1}{h} - \frac{1}{|\mathring{V}(T)| \cdot h(T)}\right) = \left(\sum\limits_{v \in \mathring{V}(T)} \frac{1}{h_v}\right) - \frac{1}{h(T)} = \sum\limits_{v \in \mathring{V}(T) \setminus \left\{\rho\right\}} \frac{1}{h_v} = B_1(T).\]
    Note that even though all binary trees have the same number of inner vertices ($n-1$), for binary trees the $B_1$ index is still induced by the third-order HM (and not the second-order HM), because disregarding the root in the summation leads to a required subtraction of the total tree height multiplied by the number $(n-1)$ of inner vertices as seen in the above definition of $\Phi_{f_{B_1}}^\mathcal{H}(T)$, so both $n$ and $h(T)$ need to be known to calculate the metaconcept's value.
\end{Rem}

Apart from the $B_1$ index, we next briefly discuss another (im)balance index based on height values.

\begin{Rem}
    The so-called \textit{maximum depth} of a tree is an imbalance index that measures the maximum depth of any of the tree's leaves, which is the height value of the root and thus the tree's height. This index was discussed in the context of phylogenetics in \cite{Colijn_phylogenetic_2014}. The maximum depth $mD$ can be expressed with the second-order height metaconcept with function
    $f_{mD}(h_v,h(T)) = \begin{cases}
        h(T), &\text{if } h_v = h(T)\\
        0, &\text{else}
    \end{cases}$.
    This is due to the root being the only inner vertex of $T$ that has the tree's height as its height value. Hence, for height values, i.e., $h_v \in [1,h(T)]$, the function $f_{mD}$ is increasing. Since the maximum depth is to the best of our knowledge one of only two (im)balance indices from the literature that are based on height values, and as it is already well-understood and has been analyzed extensively (see \cite{Fischer2023}), we focus on the $B_1$ index in the rest of this manuscript.
\end{Rem}

We now state our two main results. The first identifies the cases in which the HM induces a (binary) imbalance index. The second completely characterizes the trees that minimize the HM for all increasing functions $f$: all such trees share the height sequence of the gfb-tree.

\begin{Theo}
\label{Theo:H_imbalance_index}
    The height metaconcept $\Phi^{\mathcal{H}}_f$ induces an imbalance index for all $f$ that are strictly increasing and $1$-positive. Moreover, the height metaconcept $\Phi^{\mathcal{H}}_f$ induces a binary imbalance index for all $f$ that are strictly increasing.
\end{Theo}

Note that the proof is a direct consequence of Theorem \ref{Theo:cat_Max_H}, Parts 1 and 3, and Proposition \ref{Prop:gfb_unique_Min}, which will be stated and proven subsequently. We now turn to the complete characterization of the minimizing trees.

\begin{Theo}
\label{Theo:Min_H_gfb_characterization}
    Let $B^{min}_n$ denote the set of binary trees that minimize the height metaconcept $\Phi^{\mathcal{H}}_f$ on $\BT$ for \textit{all} (not necessarily strictly) increasing functions $f$.
    \begin{enumerate}
        \item Then $T^{gfb}_n \in B^{min}_n$, i.e., the gfb-tree minimizes the height metaconcept $\Phi^{\mathcal{H}}_f$ on $\BT$ for all increasing functions $f$. In particular
        \[\mathcal{H}\left(T^{gfb}_n\right)_i \leq \mathcal{H}\left(T\right)_i \text{ for all } i = 1, \ldots, n-1 \text{ and } T \in \BT.\]
        \item
            \begin{enumerate}[i)]
                \item For $T^{min} \in B^{min}_n$, we have $\mathcal{H}\left(T^{min}\right) = \mathcal{H}\left(T^{gfb}_n\right)$, i.e., all minimizing trees have the same height sequence as the gfb-tree and have height $\left\lceil\log_2(n)\right\rceil$, which is minimal. In particular, $\mathcal{H}(T^{min})_i = \mathcal{H}\left(T^{gfb}_n\right)_i \leq \mathcal{H}\left(T\right)_i$ for all $i = 1, \ldots, n-1$ and for all binary trees $T \in \BT$.
                \item Let $T^{min}$ be a binary tree that minimizes the height metaconcept $\Phi^{\mathcal{H}}_f$ on $\BT$ for some strictly increasing function $f$. Then $T^{min} \in B^{min}_n$, i.e., $T^{min}$ minimizes the height metaconcept for \textit{all} increasing functions $f$.
            \end{enumerate}
        \item If $T_{sub} \in \mathcal{BT}^{\ast}_{n_{sub}}$ is a pending subtree of $T^{min} \in B^{min}_n$ with $n_{sub}$ leaves, then $T_{sub} \in B^{min}_{n_{sub}}$.
    \end{enumerate}
\end{Theo}

We will prove Theorem \ref{Theo:Min_H_gfb_characterization}  at the end of Section \ref{subsubsec:complete}.

\begin{Rem}
  Note that, on $\BT$, the HM is minimized by various trees different from those minimizing the other three metaconcepts introduced in \cite{Fischer2025}, including many well-known imbalance indices such as the Sackin index, the average leaf depth, the total cophenetic index, and the (quadratic) Colless index. For example, for $n = 9$, the gfb-tree, the echelon tree, and the three trees depicted in Figure \ref{Fig:H_min_partition} are the trees minimizing the HM, whereas the imbalance indices induced by the metaconcepts mentioned above only allow (a subset of) trees having two leaf depths  $\lfloor\log_2(n)\rfloor$ and $\lceil\log_2(n)\rceil$ (see \cite{Fischer2023,Fischer2025,Cleary2025} for details).
\end{Rem}

\begin{figure}[htbp]
    \centering
    \includegraphics[width=1\textwidth]{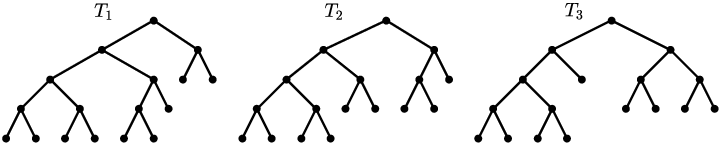}
    \caption{For $n = 9$, there are five trees, namely $T_1, T_2, T_3, T^{be}_9$ and $T^{gfb}_9$, in $\mathcal{BT}^{\ast}_9$ that minimize the HM. In particular, they all share the height sequence $\mathcal{H}(T) = (1,1,1,1,2,2,3,4)$.}
    \label{Fig:H_min_partition}
\end{figure}

In the next section, we identify cases in which the caterpillar uniquely maximizes the HM. This result is required as a first step to prove Theorem \ref{Theo:H_imbalance_index}.

\subsubsection{Maximizing trees}

In this section, we establish and prove the unique maximization of the HM by the caterpillar for all functions $f$ that satisfy the properties stated in the following theorem. This theorem plays a central role in demonstrating that the HM induces a (binary) imbalance index.

\begin{Theo}\leavevmode
\label{Theo:cat_Max_H}
\begin{enumerate}
    \item If $f$ is (strictly) increasing and $1$-positive, then $T^{cat}_n$ is the (unique) tree maximizing the height metaconcept $\Phi^{\mathcal{H}}_{f}$ on $\T$.

    \item Let $T \in \T \setminus \BT$ be a tree that is not binary. If $f$ is increasing and $1$-positive, then
    \[\Phi^{\mathcal{H}}_{f}\left(T^{cat}_n\right) > \Phi^{\mathcal{H}}_{f}(T).\]

    \item Let $T \in \BT \setminus \left\{T^{cat}_n\right\}$ be a binary tree. If $f$ is increasing, then
    \[\Phi^{\mathcal{H}}_{f}\left(T^{cat}_n\right) \geq \Phi^{\mathcal{H}}_{f}\left(T\right),\]
    where the inequality is strict for strictly increasing $f$.
\end{enumerate}
\end{Theo}

\begin{Rem}
    Notice that $f$ being $1$-positive is necessary to ensure that the caterpillar maximizes the metaconcept among all trees, not just binary ones. To see this, consider the strictly increasing but not $1$-positive function $f(x) = -\frac{1}{x}$. In this case, we have
    \[\Phi^{\mathcal{H}}_{f}\left(T^{cat}_4\right) = -\frac{1}{3} -\frac{1}{2} -1 = -\frac{11}{6} < -1 = \Phi^{\mathcal{H}}_{f}\left(T^{star}_4\right),\]
    and
    \[\Phi^{\mathcal{H}}_{f}\left(T^{fb}_2\right) = -\frac{1}{2} + 2\cdot(-1) = -\frac{5}{2} < -\frac{11}{6} = \Phi^{\mathcal{H}}_{f}\left(T^{cat}_4\right).\]
    This example shows that  for this choice of $f$ and $n = 4$, the caterpillar does not maximize the height metaconcept among all trees, but it does maximize it among all binary trees.
\end{Rem}

To prove Theorem \ref{Theo:cat_Max_H}, we will rely on the next two lemmas. The first analyzes how the height sequence is affected by a specific cherry relocation, a step that will later play a role in obtaining the caterpillar as the unique maximizing tree for the HM for certain functions $f$.

\begin{Lem}
\label{Lem:H_seq_relocate_cherry_to_cherry}
Let $n \geq 4$ and let $T \in \BT$ be a binary tree that is not the caterpillar; i.e., $T$ has at least two cherries. Let $[x,y]$ be a cherry in $T$ of maximal depth, and let $T'$ be the binary tree obtained from $T$ by relocating an arbitrary cherry $[\widetilde{x},\widetilde{y}] \neq [x,y]$ of $T$ to the leaf $x$. Then
\[\mathcal{H}(T)_i \leq \mathcal{H}(T')_i \text{ for all } i \in \left\{1,2,\ldots,n-1\right\}\]
and
\[\mathcal{H}(T)_i < \mathcal{H}(T')_i \text{ for at least one } i \in \left\{1,2,\ldots,n-1\right\}.\]
\end{Lem}
\begin{proof}
Let $n \geq 4$ and let $T \in \BT$ be a binary tree with $T \neq T^{cat}_n$. Let $[x,y]$ be a cherry in $T$ of maximal depth, and let $T'$ be the binary tree obtained from $T$ by relocating an arbitrary cherry $[\widetilde{x},\widetilde{y}] \neq [x,y]$ of $T$, which must exist as $T \neq T^{cat}_n$ and thus has more than one cherry, to the leaf $x$. Note that in $T'$, the cherry $[\widetilde{x},\widetilde{y}]$ is the unique cherry of maximal depth; for an illustration, see Figure \ref{Fig:H_relocate_cherry}.

\begin{figure}[htbp]
    \centering
    \includegraphics[width=\textwidth]{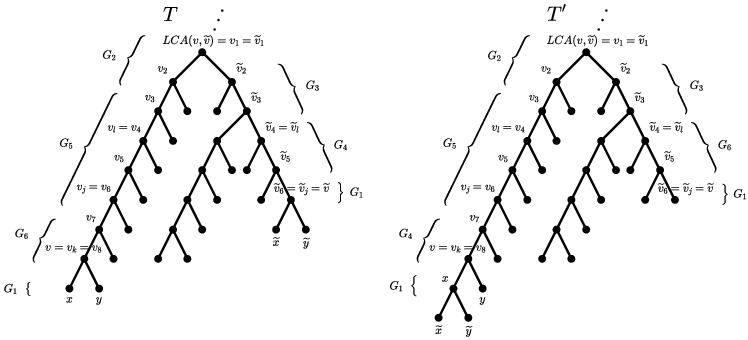}
    \caption{Example of trees $T$ and $T'$ as described in the proof of Lemma \ref{Lem:H_seq_relocate_cherry_to_cherry}. Here, $k=8$, $j=6$, and $l=4$, and the vertices are grouped into six groups $G_1, \ldots, G_6$. Recall that the cherry $[\widetilde{x},\widetilde{y}]$ can be chosen arbitrarily, as long $[\widetilde{x},\widetilde{y}] \neq [x,y]$.}
    \label{Fig:H_relocate_cherry}
\end{figure}

We now show that every entry of the height sequence of $T'$ is at least as large as the corresponding entry in $T$, with at least one entry strictly larger. This establishes the claim of the lemma.

Let $v$ be the parent of $[x,y]$ and $\widetilde{v}$ be the parent of $[\widetilde{x},\widetilde{y}]$ in $T$. Then $\delta_T(v) \geq \delta_T(\widetilde{v})$ as $[x,y]$ is a cherry of maximal depth in $T$.

Let $LCA(v,\widetilde{v}) = v_1, v_2, \ldots, v_{k-1}, v_k = v$ with $k \geq 2$ be the vertices on the path from the lowest common ancestor of $v$ and $\widetilde{v}$ to $v$ in $T$ and $T'$ (see Figure \ref{Fig:H_relocate_cherry}). Then, for $i=1, \ldots, k$, we have $h_{T}(v_i) = k-i+1$ and $h_{T'}(v_i) = k-i+2$, since $[x,y]$ is a cherry of maximal depth in $T$. Moreover, $h_{T}(x) = 0$ and $h_{T'}(x) = 1$. Thus, all height values on the path from $x$ to $v_1$ increase by one when passing from $T$ to $T'$.

On the other hand, let $LCA(v,\widetilde{v}) = \widetilde{v}_1, \widetilde{v}_2, \ldots, \widetilde{v}_{j-1}, \widetilde{v}_j = \widetilde{v}$ with $2 \leq j \leq k$ be the vertices on the path from the lowest common ancestor of $v$ and $\widetilde{v}$ to $\widetilde{v}$ in $T$ and $T'$ (see Figure \ref{Fig:H_relocate_cherry}). Then, $h_{T}(\widetilde{v}_j) = 1$ and $h_{T'}(\widetilde{v}_j) = 0$. Let $\widetilde{v}_l$ be the vertex of minimal depth whose height value decreases from $T$ to $T'$; that is, all ancestors of $\widetilde{v}_l$ (excluding $\widetilde{v}_l$) have the same height value in $T'$ as in $T$. Note that such a vertex $\widetilde{v}_l$ exists as the height value of $\widetilde{v_j}$ decreases from $1$ in $T$ to $0$ in $T'$. Since the height value of $\widetilde{v}_l$ decreases from $T$ to $T'$, the cherry $[\widetilde{x},\widetilde{y}]$ must be the unique cherry of maximal depth in $T_{\widetilde{v}_l}$ (and thus in all pending subtrees of all descendants of $\widetilde{v}_l$), otherwise the height of $T_{\widetilde{v}_l}$ would remain unchanged. Consequently, the height values of all descendants of $\widetilde{v}_l$ decrease by one after the cherry relocation. For an example see Figure \ref{Fig:H_relocate_cherry}, where $\widetilde{v}_l = \widetilde{v}_4$.

Note that $l$ can take any value between $2$ and $j$. Thus, $h_{T}(\widetilde{v}_i) = h_{T'}(\widetilde{v}_i)$ for $2 \leq i \leq l-1$. If $l \neq j$, the exact height values of the vertices $\widetilde{v}_{j-1}, \widetilde{v}_{j-2}, \ldots, \widetilde{v}_{l}$ have not been given yet. In the case that $l \neq j$, these vertices attain in $T$ the height values $2, 3, \ldots, j-l+1$ and in $T'$ the height values $1, 2, \ldots, j-l$.

We  summarize all height values in $T$ and $T'$ in Table \ref{Tab:HeightChange}, grouping the vertices into six groups, $G_1, \ldots, G_6$, corresponding to the notation in Figure \ref{Fig:H_relocate_cherry}.

\begin{table}[htbp]
    \small
	\centering
	\caption{Overview of the different groups of vertices described in the proof of Lemma~\ref{Lem:H_seq_relocate_cherry_to_cherry}. Note that some groups do not exist for certain combinations of $j,k,l$. For example, $G_3$ does not exist if $l = 2$ as there, the indices range from $2$ to $l-1$.}
	\label{Tab:HeightChange}
	\begin{tabular}{ c | l || c | l}
	\textbf{Group} & \textbf{Height values} & \textbf{Group} & \textbf{Height values}\\
	\hline \hline
    $G_1$ &  $\begin{aligned}
            h_{T}(x) &= 0 = h_{T'}(\widetilde{v}_j)\\
            h_{T}(\widetilde{v}_j) &= 1 = h_{T'}(x)
        \end{aligned}$

    &$G_4$ & $\begin{aligned}
            h_{T}(\widetilde{v}_l) &= j-l+1 = h_{T'}(v_{k-j+l+1})\\
            h_{T}(\widetilde{v}_{l+1}) &= j-l = h_{T'}(v_{k-j+l+2})\\
            &\vdots\\
            h_{T}(\widetilde{v}_{j-2}) &= 3 = h_{T'}(v_{k-1})\\
            h_{T}(\widetilde{v}_{j-1}) &= 2 = h_{T'}(v_k)\\
        \end{aligned}$ \\
    \hline
    $G_2$ & $\begin{aligned}
            h_{T}(v_1) &= k < k+1 = h_{T'}(v_1)\\
            h_{T}(v_2) &= k-1 < k = h_{T'}(v_2)
        \end{aligned}$

    &$G_5$ & $\begin{aligned}
            h_{T}(v_3) &= k-2 < k-1 = h_{T'}(v_3)\\
            h_{T}(v_4) &= k-3 < k-2 = h_{T'}(v_4)\\
            &\vdots\\
            h_{T}(v_{k-j+l}) &= j-l+1 < j-l+2 = h_{T'}(v_{k-j+l})\\
        \end{aligned}$ \\
    \hline 
    $G_3$ & $\begin{aligned}
            h_{T}(\widetilde{v}_2) &= h_{T'}(\widetilde{v}_2)\\
            h_{T}(\widetilde{v}_3) &= h_{T'}(\widetilde{v}_3)\\
            &\vdots\\
            h_{T}(\widetilde{v}_{l-1}) &= h_{T'}(\widetilde{v}_{l-1})\\
        \end{aligned}$

    &$G_6$ & $\begin{aligned}
            h_{T}(v_{k-j+l+1}) &= j-l = h_{T'}(\widetilde{v}_l)\\
            h_{T}(v_{k-j+l+2}) &= j-l-1 = h_{T'}(\widetilde{v}_{l+1})\\
            &\vdots\\
            h_{T}(v_k) &= 1 = h_{T'}(\widetilde{v}_{j-1})\\
        \end{aligned}$
\end{tabular}
\end{table}

As we can see, for every height value in $T$, there exists a height value in $T'$ (possibly corresponding to a different vertex) that is at least as large. Importantly, the height value of the lowest common ancestor of $v$ and $\widetilde{v}$, $LCA(v,\widetilde{v}) = v_1 = \widetilde{v}_1$, and the height value of $v_2$ both increase strictly by one.

Furthermore, the height values of all vertices on the path from $LCA(v,\widetilde{v})$ to the root $\rho$ (if any) increase by one. All other vertices are not affected by the relocation.

Taken together, when comparing the height sequences, for every height value in $T$ there exists a height value in $T'$ that is at least as large, and at least one value in $T'$ is strictly larger than the corresponding value in $T$. This completes the proof.
\end{proof}

Next, we consider the reverse cherry relocation of the one described in the previous lemma. This corollary will be needed later on to analyze the minimizing trees.

\begin{Cor}
\label{Cor:H_seq_relocate_maxcherry_to_leaf}
    Let $T' \in \BT$ be a binary tree with the following two properties: (1) $T'$ has a unique cherry $[\widetilde{x},\widetilde{y}]$ of maximal depth $k \geq 3$, and (2) there exists a leaf $l \neq \widetilde{x},\widetilde{y}$ of depth $\delta_l \leq k-2$ in $T'$. Let $T$ be the binary tree obtained from $T'$ by relocating $[\widetilde{x},\widetilde{y}]$ to $l$. Then
    \[\mathcal{H}(T)_i \leq \mathcal{H}(T')_i \text{ for all } i \in \left\{1,2,\ldots,n-1\right\}\]
    and
    \[\mathcal{H}(T)_i < \mathcal{H}(T')_i \text{ for at least one } i \in \left\{1,2,\ldots,n-1\right\}.\]
\end{Cor}
\begin{proof}
    From the fact that there exists a unique cherry of maximal depth $k \geq 3$, it follows immediately that $n \geq 4$. The proof then follows directly from Lemma \ref{Lem:H_seq_relocate_cherry_to_cherry}, since the cherry relocation described here is precisely the reverse of the relocation considered in that lemma. This is because, in Lemma \ref{Lem:H_seq_relocate_cherry_to_cherry}, the relocated cherry is chosen arbitrarily, and after its relocation, it becomes the unique cherry of maximal depth.
\end{proof}

Next, we present another lemma needed for the proof of Theorem \ref{Theo:cat_Max_H}. In this lemma, we compare the height sequence of the caterpillar with that of an arbitrary tree. We sort the height sequence in descending order, which simplifies the notation when the other tree is not binary.

\begin{Lem}
\label{Lem:H_seq_caterpillar}
    Let $\mathcal{H}^d(T)$ denote the height sequence of a tree $T$ in descending  (instead of ascending) order. For all trees $T \in \T \setminus \left\{T^{cat}_n\right\}$, we have
    \[\mathcal{H}^d(T)_i \leq \mathcal{H}^d\left(T^{cat}_n\right)_i \text{ for all } i \in \left\{1,2,\ldots,|\mathring{V}(T)|\right\}.\]
     If $T$ is binary, then
    \[\mathcal{H}^d(T)_i < \mathcal{H}^d\left(T^{cat}_n\right)_i \text{ for at least one } i \in \left\{1,2,\ldots,n-1\right\}.\]
\end{Lem}

We remark that the proof of Lemma~\ref{Lem:H_seq_caterpillar} follows the same principle as that of \cite[Proposition 3.15]{Fischer2025}.

\begin{proof}[Proof of Lemma~\ref{Lem:H_seq_caterpillar}]
    First, assume that $T \in \T \setminus \BT$. We now transform $T$ into a binary tree without decreasing the height values of its vertices. By the choice of $T$, there exists an inner vertex with at least $3$ children; denote one such vertex by $u_1$, and let its children be $v_1, \ldots, v_k$ with $k \geq 3$. We construct a tree $T'$ from $T$ as follows: delete the edges $(u_1,v_i)$ for all $2 \leq i \leq k$, add a new vertex $u_0$, and add the edges $(u_1,u_0)$ and $(u_0,v_i)$ for all $2 \leq i \leq k$. For an illustration, see Figure \ref{Fig:H_caterpillar_T_to_binary}. Repeating this procedure iteratively yields a binary tree.

    Comparing the height values of the vertices of $T$ and $T'$, we observe that in $T'$ either $u_1$ or $u_0$ has the same height value as $u_1$ in $T$. Moreover, all other vertices either retain their height value or have their height value increased by one; in particular, no vertex decreases its height value.

    Thus, it remains to prove the statement for binary trees.

    \begin{figure}[htbp]
        \centering
        \includegraphics[scale=2]{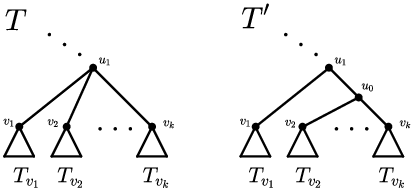}
        \caption{$T$ and $T'$ as described in the proof of Lemma \ref{Lem:H_seq_caterpillar} and Theorem \ref{Theo:B1_extrema}.}
        \label{Fig:H_caterpillar_T_to_binary}
    \end{figure}

    If $T \in \BT \setminus \{T^{cat}_n\}$, then $n \geq 4$, and we can transform $T$ into the caterpillar by repeatedly relocating a cherry of $T$ to a leaf of a cherry with maximal depth, as described in Lemma \ref{Lem:H_seq_relocate_cherry_to_cherry}. The statement then follows directly from Lemma \ref{Lem:H_seq_relocate_cherry_to_cherry}.
\end{proof}

We are now in a position to prove Theorem \ref{Theo:cat_Max_H}.

\begin{proof}[Proof of Theorem \ref{Theo:cat_Max_H}]\leavevmode
    There are three statements to prove. We first consider the second and third. 
    \begin{enumerate}
        \item[2.] Let $T \in \T \setminus \BT$; in particular, $n \geq 3$. Then $T$ has $|\mathring{V}(T)| < n-1$ inner vertices. Moreover, let $f$ be an increasing and $1$-positive function, and let $\mathcal{H}^d$ be the height sequence of a tree in descending order. Then,
        \begin{align*}
            \Phi^{\mathcal{H}}_{f}\left(T^{cat}_n\right) &= \sum\limits_{i = 1}^{|\mathring{V}(T)|} f\left(\mathcal{H}^d\left(T^{cat}_n\right)_i\right) + \sum\limits_{i = |\mathring{V}(T)|+1}^{n-1} f\left(\mathcal{H}^d\left(T^{cat}_n\right)_i\right)\\
            &\stackrel{f \text{ $1$-pos.}}{>} \sum\limits_{i = 1}^{|\mathring{V}(T)|} f\left(\mathcal{H}^d\left(T^{cat}_n\right)_i\right) \geq \sum\limits_{i = 1}^{|\mathring{V}(T)|} f\left(\mathcal{H}^d(T)_i\right) = \Phi^{\mathcal{H}}_{f}(T),
        \end{align*}
        where the last inequality follows from Lemma \ref{Lem:H_seq_caterpillar} and the fact that $f$ is increasing.
    
        \item[3.] Now, let $f$ be a (strictly) increasing function and let $T \in \BT \setminus \{T^{cat}_n\}$ be a binary tree with $n$ leaves. In this case, the statement follows directly from Lemma \ref{Lem:Seq_entries_Min_Max_meta_binary}, Part 2(b), and Lemma \ref{Lem:H_seq_caterpillar}.
    
        \item[1.] This statement is a direct consequence of the two previous parts.
    \end{enumerate}
    This completes the proof.
\end{proof}

Next, we analyze which trees minimize the HM on $\T$ and $\BT$, respectively.

\subsubsection{Complete characterization of all trees minimizing the HM \texorpdfstring{$\Phi^{\mathcal{H}}_f$}{PhiHf}} \label{subsubsec:complete}

In this section, we first discuss the minimization of the HM on $\T$ and then turn to the minimization on $\BT$. We begin by identifying general properties of minimizing trees, which will be needed to prove Theorem \ref{Theo:Min_H_gfb_characterization}. This theorem provides a complete characterization of the trees minimizing the HM for all increasing functions $f$: the gfb-tree is always a minimizer, and every other minimizing tree has the same height sequence as the gfb-tree. We then focus on the number of minimizing trees. In particular, we show that there is a unique minimizing tree (namely, the gfb-tree) for the HM if and only if the leaf number $n$ can be expressed as a difference of two powers of two, i.e., $n = 2^h - 2^i$ for some integers $h > i \geq 0$. Finally, we conclude the analysis on $\BT$ by presenting several examples of minimizing trees other than the gfb-tree.

We begin with the unique minimization of the star tree on $\T$.

\begin{Prop}
\label{Prop:star_Min_on_T}
Let $T^{star}_n$ be the star tree on $n$ leaves and let $f$ be a $1$-positive (but not necessarily increasing) function. Then the star tree is the unique tree minimizing the height metaconcept $\Phi^{\mathcal{H}}_f$ on $\T$.
\end{Prop}
\begin{proof}
Let $T^{star}_n$ be the star tree on $n$ leaves and let $f$ be a $1$-positive function. For $n \leq 2$, there exists only one tree, so there is nothing to show. Now let $n \geq 3$. Moreover, let $T \in \T \setminus \left\{T^{star}_n\right\}$ be another tree on $n$ leaves, i.e., $|\mathring{V}(T)| \geq 2$. 

Every tree with at least two leaves has at least one inner vertex whose children are all leaves, i.e., an inner vertex with height value $1$. Therefore,
\[\Phi^{\mathcal{H}}_f\left(T^{star}_n\right) = f(1) \stackrel{f \text{ $1$-pos.}}{<} f(1) + \sum\limits_{i = 2}^{|\mathring{V}(T)|} f\left(\mathcal{H}(T)_i\right) = \Phi^{\mathcal{H}}_f(T).\] 
This completes the proof.
\end{proof}

In what follows, we analyze the binary minimizing trees. Our goal is to prove Theorem \ref{Theo:Min_H_gfb_characterization}, which states that all trees minimizing the HM for all increasing functions $f$ share the same height sequence as the gfb-tree. To establish this result, we first describe several properties that any minimizing binary tree must satisfy. We then present constructions that allow one to obtain minimizing trees with a given number of leaves $n$ from minimizing trees with another number of leaves $m$. These constructions will play a central role in the proof of Theorem \ref{Theo:Min_H_gfb_characterization}, where we use them to obtain a contradiction.

The next lemma plays a crucial role in several of the upcoming proofs, as it transfers a key property of the gfb-tree (see Remark \ref{Rem:gfb_properties}) to all minimizing trees.

\begin{Lem}
\label{Lem:T_Min_H_properties}
Let $T \in \BT$ be a binary tree with $n$ leaves that minimizes the height metaconcept $\Phi^{\mathcal{H}}_f$ for some strictly increasing function $f$. Then:
\begin{itemize}
    \item If $n$ is even, all leaves of $T$ are part of a cherry. Moreover, if $n \geq 4$, then $T$ contains $T^{fb}_2$ as a pending subtree.
    \item If $n$ is odd, all leaves except for one are part of a cherry. Moreover, if $n \geq 7$, then $T$ contains $T^{fb}_2$ as a pending subtree.
\end{itemize}
\end{Lem}
\begin{proof}
For each $n \leq 3$, there exists only one tree with $n$ leaves, and it is easily checked that in these cases, the statement is true.

Now, let $T$ be a binary tree with $n \geq 4$ leaves that minimizes the HM $\Phi^{\mathcal{H}}_f$ for some strictly increasing function $f$.

First, we show that $T$ contains at most one leaf that is not part of a cherry. Consequently, if $n$ is even, all leaves of $T$ are part of a cherry, and if $n$ is odd, $T$ has exactly one leaf that is not part of a cherry. The strategy of the proof is to show that for any tree $T$ with at least two leaves not belonging to a cherry, there exists another tree $T'$ such that $\Phi^{\mathcal{H}}_f(T) > \Phi^{\mathcal{H}}_f(T')$, thus leading to a contradiction to the minimality of $T$.

Hence, assume that $T$ has two leaves, say $x$ and $y$, that are not part of a cherry. Let $v_x$ and $v_y$ denote the parents of $x$ and $y$, respectively. We distinguish two cases depending on whether one of $v_x$ and $v_y$ is an ancestor of the other.

\begin{enumerate}
    \item We first consider the case in which one of $v_x$ and $v_y$ is ancestral to the other one. Without loss of generality, assume that $v_x$ is an ancestor of $v_y$. This implies that the depth of $v_x$ is smaller than the depth of $v_y$. Let $w_y$ be the other child of $v_y$; since $y$ is not part of a cherry, $w_y$ is an inner vertex. Let $T'$ be the tree obtained by swapping $T_{w_y}$ and $x$ (see Figure \ref{Fig:leaves_in_cherries_1case} for an illustration). Then at least the height value of $v_y$ decreases from $T$ to $T'$. Moreover, all ancestors of $v_y$ may also decrease their height values, but no vertex increases its height value. Since $f$ is strictly increasing, we conclude that $\Phi^{\mathcal{H}}_f(T) > \Phi^{\mathcal{H}}_f(T')$.

    \begin{figure}[htbp]
        \centering
        \includegraphics[scale=1]{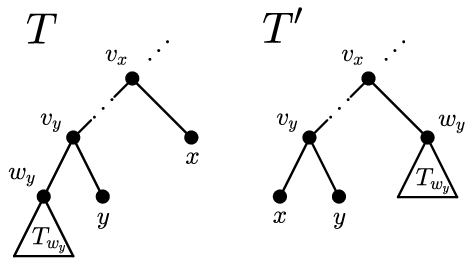}
        \caption{Illustration of $T$ and $T'$ as used in the first case in the proof of Lemma \ref{Lem:T_Min_H_properties}.}
        \label{Fig:leaves_in_cherries_1case}
    \end{figure}

    \item Now, assume that $v_x$ and $v_y$ are not ancestors of one another. Let $w_y$ and $w_x$ denote the other children of $v_y$ and $v_x$, respectively. Without loss of generality, assume that $h_{w_x} \geq h_{w_y}$. Let $T'$ be the tree obtained from $T$ by swapping $T_{w_y}$ and $x$ (see Figure \ref{Fig:leaves_in_cherries_2case}). Again, at least the height value of $v_y$ decreases, and all of its ancestors may decrease their height values as well. All other height values are not affected, since $h_{w_x} \geq h_{w_y}$. Hence, because $f$ is strictly increasing, we have $\Phi^{\mathcal{H}}_f(T) > \Phi^{\mathcal{H}}_f(T')$.

    \begin{figure}
        \centering
        \includegraphics[scale=1]{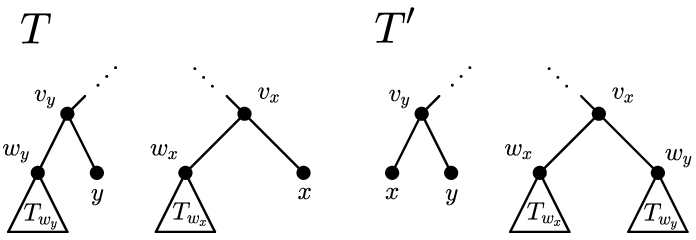}
        \caption{Illustration of $T$ and $T'$ as used in the second case in the proof of Lemma \ref{Lem:T_Min_H_properties}. Note that $h_{w_x} \geq h_{w_y}$.}
        \label{Fig:leaves_in_cherries_2case}
    \end{figure}
\end{enumerate}
As both cases contradict the minimality of $T$, we can conclude that $T$ has at most one leaf that is not part of a cherry. 

Next, we show that $T$ must contain $T^{fb}_2$ as a pending subtree if $n \geq 4$ is even. Note that for $n\geq 4$, we have $h(T)\geq 2$. Now, let $v$ be the grandparent (i.e., the parent of the parent) of a cherry of maximal depth, and let $v'$ denote the other child of $v$ (i.e., not the cherry parent). By the first part of the proof, $T_{v'}$ cannot solely consist of a single leaf (as $n$ is even), i.e., $T_{v'}$ cannot have height $0$. Moreover, $T_{v'}$ cannot have height greater than $1$, since this would contradict the maximal depth of the cherry. Therefore, $T_{v'}$ has height $1$, i.e., it is the parent of a cherry. Consequently, $T_v = T^{fb}_2$, and thus $T$ contains $T^{fb}_2$ as a pending subtree. This completes the proof for the case when $n$ is even.

Next, we show that $T$ contains $T^{fb}_2$ as a pending subtree if $n \geq 7$ is odd. From the first part of the proof, we know that $T$ has a unique leaf, say $z$, which is not part of a cherry. Let $v$ be the parent of $z$, and let $T'_v$ be the pending subtree rooted in the other child of $v$. Suppose $T'_v$ has $n'$ leaves. Note that $n' \geq 2$ (because $n'=1$ would imply that $z$ belongs to a cherry; a contradiction), and moreover,  $n'$ must be even; otherwise $T'_v$ also would contain a leaf that is not part of a cherry, contradicting the uniqueness of $z$ in $T$. We now distinguish the two cases: $n' = 2$ and $n' > 2$.

\begin{enumerate}[(i)]
    \item First, assume that $n' > 2$; since $n'$ is even, this means $n' \geq 4$. In this case, $T'_v$, and therefore also $T$, contains $T^{fb}_2$ as a pending subtree by the previous part of the proof.

    \item Now, assume $n' = 2$. Then $T_v = T^{cat}_3$. The remaining $n-3 \geq 4$ leaves appear in cherries as otherwise this would contradict the uniqueness of $z$. If one pending subtree of $T$ that does not contain $T_v$ has at least $4$ leaves, we note that the number of leaves in this subtree must be even (as otherwise there would be a second leaf not contained in a cherry), which is why we can then apply the previous part of the proof to conclude that this subtree must contain $T_2^{fb}$. It now only remains to consider the case in which no pending subtree of $T$ which does \textit{not} contain $T_v$ has at least four leaves. In this case, \textit{all} pending subtrees of $T$ not containing $T_v$ must contain precisely two leaves (because again, they have to contain an even number of leaves as $T$ has only one leaf not contained in a cherry). In this case, $T$ looks as follows (cf. Figure \ref{Fig:T_Min_H_properties}): All vertices on the path from the parent of $v$ to the root of $T$ have a cherry parent as a child, and the other child contains $T_v$ in its pending subtree. In particular, the cherry in $T_v$ is the unique cherry of maximal depth $k \geq 4$ of $T$, as $T$ has at least seven leaves, so the path from $v$ to the root of $T$ has at least length $2$ (as there are at least two more cherry subtrees descending from $v$'s path to the root), and as the depth of the cherry leaves of $T_v$ in $T$ is said path length plus $2$, the depth of these leaves is strictly larger than that of all other leaves in the tree. Moreover, the leaves in the maximal pending subtree of $T$ that consists only of a cherry have depth $2 \leq k-2$. Let $\widehat{T}$ be the tree obtained from relocating the cherry of maximal depth to one of these leaves (cf. Figure \ref{Fig:T_Min_H_properties}). By Corollary \ref{Cor:H_seq_relocate_maxcherry_to_leaf}, and since $f$ is strictly increasing, we have $\Phi^{\mathcal{H}}_f(\widehat{T}) < \Phi^{\mathcal{H}}_f(T)$. This contradicts the minimality of $T$. This contradiction shows that our assumption was wrong, i.e., $T$ cannot be such that it contains no pending $T^{fb}_2$ subtree. This completes the proof. 
\end{enumerate}
\end{proof}

\begin{figure}
    \centering
    \includegraphics[scale=1.5]{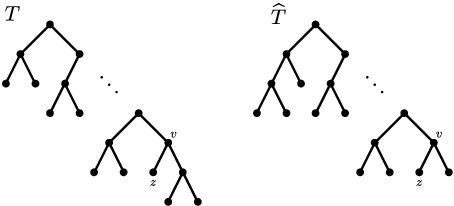}
    \caption{The only possible tree $T$ in the last part of the proof of Lemma \ref{Lem:T_Min_H_properties} for the case $n \geq 7$ and $n' = 2$ that does not contain $T^{fb}_2$ as a  pending subtree. Tree $\widehat{T}$ is obtained from $T$ by the cherry relocation described in the same part of the proof.}
    \label{Fig:T_Min_H_properties}
\end{figure}

In the following, we demonstrate how to construct all binary minimizing trees with an odd number of leaves from those with the next higher even number of leaves, and vice versa. These constructions will be needed for the proof of Theorem \ref{Theo:Min_H_gfb_characterization}, i.e., the complete characterization of the minimizing trees, where they are used to derive a contradiction.

\begin{Lem}
\label{Lem:Min_n-1_n_sameHplus1}
Let $f$ be strictly increasing and let $n \geq 4$ be even. Let $T_{n-1} \in \mathcal{BT}^{\ast}_{n-1}$ and $T_n \in \BT$. Assume that
\[\mathcal{H}(T_n)_{i+1} = \mathcal{H}(T_{n-1})_i \text{ for all } i = 1, \ldots, n-2.\]
Then $T_{n-1}$ minimizes the height metaconcept $\Phi^{\mathcal{H}}_f$ on $\mathcal{BT}^{\ast}_{n-1}$ if and only if $T_n$ minimizes the height metaconcept $\Phi^{\mathcal{H}}_f$ on $\BT$.
\end{Lem}

\begin{figure}
    \centering
    \includegraphics[width=0.5\linewidth]{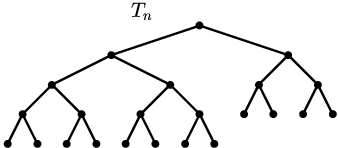}\\
    \includegraphics[width=\textwidth]{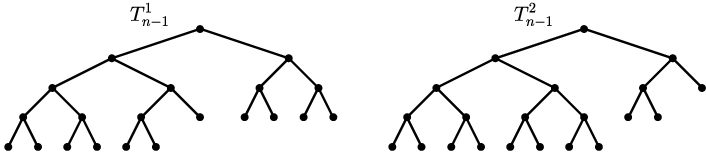}
    \caption{Both pairs $(T_n, T_{n-1}^1)$ and $(T_n, T_{n-1}^2)$ satisfy the conditions of Lemma \ref{Lem:Min_n-1_n_sameHplus1}. Note that $T^1_{n-1}$ and $T^2_{n-1}$ can be obtained by deleting different cherries of $T_n$. We have $\mathcal{H}(T_n) = (1,\underline{1,1,1,1,1,2,2,2,3,4})$ and $\mathcal{H}(T^1_{n-1}) = (1,1,1,1,1,2,2,2,3,4) = \mathcal{H}(T^2_{n-1})$, where the latter sequence coincides with the underlined part of $\mathcal{H}(T_n)$, as stated by Lemma \ref{Lem:Min_n-1_n_sameHplus1}.}
    \label{Fig:Tn_Tn-1}
\end{figure}

\begin{proof}
Let $f$ be strictly increasing and let $n \geq 4$ be even. Let $T_{n-1}$ and $T_n$ be as described (for an example, see Figure \ref{Fig:Tn_Tn-1}). In particular,
\[\mathcal{H}(T_n)_{i+1} = \mathcal{H}(T_{n-1})_i \text{ for all } i = 1, \ldots, n-2.\]
We prove both directions by contradiction.

Suppose that $T_{n-1}$ minimizes the HM on $\mathcal{BT}^{\ast}_{n-1}$. Seeking a contradiction, assume that $T_n$ does not minimize the HM on $\BT$, but that some tree $T'_n$ does. In other words, $\Phi^{\mathcal{H}}_f(T_n) > \Phi^{\mathcal{H}}_f(T'_n)$. By Lemma \ref{Lem:T_Min_H_properties},  tree $T'_n$ contains $T^{fb}_2$ as a pending subtree. Let $T'_{n-1}$ be the tree obtained from $T'_n$ by deleting a cherry of the pending subtree $T^{fb}_2$. By construction,
\[\mathcal{H}(T'_n)_{i+1} = \mathcal{H}(T'_{n-1})_i \text{ for all } i = 1, \ldots, n-2.\]
Moreover, $\mathcal{H}(T_n)_1 = 1 = \mathcal{H}(T'_n)_1$. Consequently,
\[\Phi^{\mathcal{H}}_f(T_{n}) = f(1) + \Phi^{\mathcal{H}}_f(T_{n-1}) \leq f(1) + \Phi^{\mathcal{H}}_f(T'_{n-1}) = \Phi^{\mathcal{H}}_f(T'_{n}),\]
where the inequality holds as $T_{n-1}$ minimizes the HM by assumption. This contradicts the assumption that $\Phi^{\mathcal{H}}_f(T_n) > \Phi^{\mathcal{H}}_f(T'_n)$. Therefore, $T_n$  must minimize the HM on $\BT$. \smallskip

The converse direction works similarly: Suppose that $T_{n}$ minimizes the HM on $\BT$. Seeking a contradiction, assume that $T_{n-1}$ does not minimize the HM on $\mathcal{BT}^{\ast}_{n-1}$, but that some tree $T'_{n-1}$ does. In other words, $\Phi^{\mathcal{H}}_f(T_{n-1}) > \Phi^{\mathcal{H}}_f(T'_{n-1})$. By Lemma \ref{Lem:T_Min_H_properties} and since $n-1$ is odd, tree $T'_{n-1}$ has a unique leaf that is not part of a cherry. Let $T'_{n}$ be the tree obtained from $T'_{n-1}$ by attaching a cherry to this leaf. By construction,
\[\mathcal{H}(T'_n)_{i+1} = \mathcal{H}(T'_{n-1})_i \text{ for all } i = 1, \ldots, n-2.\]
Moreover, $\mathcal{H}(T_n)_1 = 1 = \mathcal{H}(T'_n)_1$. Thus,
\[\Phi^{\mathcal{H}}_f(T_{n-1}) = \Phi^{\mathcal{H}}_f(T_{n}) - f(1) \leq \Phi^{\mathcal{H}}_f(T'_{n}) - f(1) = \Phi^{\mathcal{H}}_f(T'_{n-1}).\]
This contradicts the assumption that $\Phi^{\mathcal{H}}_f(T_{n-1}) > \Phi^{\mathcal{H}}_f(T'_{n-1})$. Therefore, $T_{n-1}$ must  minimize the HM on $\mathcal{BT}^\ast_{n-1}$.

This completes the proof.
\end{proof}

We now analyze how every binary minimizing tree with an even number of leaves $n$ can be constructed from binary minimizing trees with $n-1$ leaves. This construction transfers a property the gfb-tree has (see Remark \ref{Rem:gfb_properties}) to all minimizing trees.

\begin{Prop}
\label{Prop:H_Min_n-1_to_Min_n}
Let $n \geq 4$ be even, and let $f$ be strictly increasing. Denote by $B_{n}$ the set of binary trees with $n$ leaves that are obtained as follows: take any binary tree $T_{n-1}$ on $n-1$ leaves that minimizes the height metaconcept $\Phi^{\mathcal{H}}_f$ on $\mathcal{BT}^{\ast}_{n-1}$, and attach a cherry to the unique leaf of $T_{n-1}$ that is not part of a cherry. Let $B^{min}_{n}$ be the set of trees that minimize the HM $\Phi^{\mathcal{H}}_f$ on $\BT$. Then $B_{n} = B^{min}_{n}$.
\end{Prop}

Note that if there is a unique tree minimizing the HM on $\mathcal{BT}^{\ast}_{n-1}$, then $|B_n| = |B^{min}_{n}| = 1$. In other words, in this case there is also a unique tree minimizing the HM on $\BT$.

\begin{proof}
Let $f$ be strictly increasing.

First, we show that $B_{n} \subseteq B^{min}_{n}$. Let $T_{n-1}$ be a binary tree with an odd number $n-1 \geq 3$ of leaves that minimizes the HM. By Lemma \ref{Lem:T_Min_H_properties}, $T_{n-1}$ has a unique leaf that is not part of a cherry. Let $T_n$ be the tree obtained by attaching a cherry to this leaf, so that $T_n \in B_n$. Then
\[\mathcal{H}(T_n)_{i+1} = \mathcal{H}(T_{n-1})_i \text{ for all } i = 1, \ldots, n-2.\]
By Lemma \ref{Lem:Min_n-1_n_sameHplus1}, it follows that $T_n$ minimizes the HM, i.e., $T_n \in B^{min}_n$.

Second, we show that $B_{n} \supseteq B^{min}_{n}$. Let $T_n$ be a binary tree with an even number of leaves $n \geq 4$ that minimizes the HM, i.e., $T_n \in B^{min}_{n}$. By Lemma \ref{Lem:T_Min_H_properties}, all leaves of $T_n$ are part of a cherry. Let $T_{n-1}$ be any of the binary trees with $n-1$ leaves obtained from $T_n$ by deleting an arbitrary cherry $[x,y]$. We need to show that $T_{n-1}$ minimizes the HM, which will imply that $T_n \in B_n$. Let $v$ be the grandparent of the cherry $[x,y]$, i.e., the parent of its cherry parent. Then, the other maximal pending subtree of $T_v$, say $T'$, not containing $[x,y]$, has height $h(T') \geq 1$, i.e., deleting the cherry $[x,y]$ does not change the height value of $v$ and hence of no other vertex. This is because if $h(T') = 0$, $T'$ is a leaf not part of a cherry, a contradiction. Hence,  by construction, we have
\[\mathcal{H}(T_n)_{i+1} = \mathcal{H}(T_{n-1})_i \text{ for all } i = 1, \ldots, n-2.\]
Now, by Lemma \ref{Lem:Min_n-1_n_sameHplus1}, $T_{n-1}$ minimizes the HM. Therefore, $T_n \in B_n$. This completes the proof.
\end{proof}

Next, we prove the converse direction, i.e., we analyze how all binary trees minimizing the HM with $n-1$ leaves can be constructed from the set of binary trees minimizing the HM with $n$ leaves, where $n$ is even. This again generalizes a property of the gfb-tree (see Remark \ref{Rem:gfb_properties}).

\begin{Prop}
\label{Prop:H_Min_n_to_Min_n-1}
Let $n \geq 4$ be even, and let $f$ be strictly increasing. Denote by $B_{n-1}$ the set of binary trees with $n-1$ leaves that can be obtained as follows: take an arbitrary binary tree $T_{n}$ on $n$ leaves that minimizes the height metaconcept $\Phi^{\mathcal{H}}_f$ on $\BT$, and delete an arbitrary cherry of $T_n$. Let $B^{min}_{n-1}$ be the set of trees with $n-1$ leaves that minimize the height metaconcept $\Phi^{\mathcal{H}}_f$ on $\mathcal{BT}^{\ast}_{n-1}$. Then $B_{n-1} = B^{min}_{n-1}$.
\end{Prop}

The following example shows that, in contrast to Proposition \ref{Prop:H_Min_n-1_to_Min_n}, there may be multiple trees minimizing the HM on $\mathcal{BT}^{\ast}_{n-1}$ even in cases in which there is a unique tree minimizing the HM on $\BT$.

\begin{Ex}
    Let $f$ be a strictly increasing function. First, we show that there are two trees minimizing the HM on $\mathcal{BT}^{\ast}_5$, where $|\mathcal{BT}^{\ast}_5| = 3$, i.e., there are three binary trees with five leaves: $T^{gfb}_5 = \left(T^{cat}_3,T^{fb}_1\right)$, $T^{be}_5 = \left(T^{fb}_2,T^{fb}_0\right)$ (see Figure \ref{Fig:5_2_3} in the appendix), and $T^{cat}_5$. By Theorem \ref{Theo:cat_Max_H}, Part 3, the caterpillar uniquely maximizes the HM. Moreover, the height sequences of the other two trees are $\mathcal{H}\left(T^{gfb}_5\right) = (1,1,2,3) = \mathcal{H}\left(T^{be}_5\right)$ as illustrated in Figure \ref{Fig:5_2_3} in the appendix. Hence, both $T^{gfb}_5$ and $T^{be}_5$ minimize the HM on $\mathcal{BT}^{\ast}_{5}$.

    Next, we show that there is a unique tree minimizing the HM on $\mathcal{BT}^{\ast}_{6}$. By Proposition \ref{Prop:H_Min_n-1_to_Min_n}, all minimizing trees on $\mathcal{BT}^{\ast}_{6}$ can be obtained by attaching a cherry to the unique leaf of $T^{gfb}_{5}$ and $T^{be}_5$ that is not part of a cherry. In both cases, this operation yields the gfb-tree on six leaves, i.e., $T^{gfb}_6$, as depicted in Figure \ref{Fig:6_5_6} in the appendix, which is therefore the unique tree minimizing the HM on $\mathcal{BT}^{\ast}_{6}$.
\end{Ex}

\begin{proof}[Proof of Proposition \ref{Prop:H_Min_n_to_Min_n-1}]
Let $f$ be strictly increasing.

First, we show that $B_{n-1} \subseteq B^{min}_{n-1}$. Let $T_{n}$ be a binary tree with an even number of leaves $n \geq 4$ that minimizes the HM. Let $T_{n-1}$ be the tree obtained from $T_n$ by deleting an arbitrary cherry of $T_n$, so that $T_{n-1} \in B_{n-1}$. We need to prove that $T_{n-1}$ minimizes the HM on $\mathcal{BT}^{\ast}_{n-1}$, i.e., $T_{n-1} \in B^{min}_{n-1}$.

By Lemma \ref{Lem:T_Min_H_properties}, all leaves of $T_n$ are part of a cherry. Let $[x,y]$ be an arbitrary cherry of $T_n$ with parent $v$. As in the proof of Proposition \ref{Prop:H_Min_n-1_to_Min_n}, the pending subtree of the other child of the parent of $v$ has at least height value $1$, i.e., it is not a single leaf. Thus,
\[\mathcal{H}(T_n)_{i+1} = \mathcal{H}(T_{n-1})_i \text{ for all } i = 1, \ldots, n-2.\]
By Lemma \ref{Lem:Min_n-1_n_sameHplus1}, $T_{n-1}$ minimizes the HM, and therefore $T_{n-1} \in B^{min}_{n-1}$. \medskip

Second, we show that $B_{n-1} \supseteq B^{min}_{n-1}$. Let $T_{n-1}$ be a binary tree with an odd number of leaves $n \geq 3$ that minimizes the HM, i.e., $T_{n-1} \in B^{min}_{n-1}$. By Lemma \ref{Lem:T_Min_H_properties}, $T_{n-1}$ has a unique leaf $x$ that is not part of a cherry. Thus, $T_{n-1}$ can be obtained from the binary tree $T_n$ on $n$ leaves that has a cherry in place of the leaf $x$ by deleting that cherry.

We need to prove that $T_{n}$ minimizes the HM, i.e., $T_n \in B^{min}_{n}$, which, by definition, implies $T_{n-1} \in B_{n-1}$. By construction,
\[\mathcal{H}(T_n)_{i+1} = \mathcal{H}(T_{n-1})_i \text{ for all } i = 1, \ldots, n-2.\]
By Lemma \ref{Lem:Min_n-1_n_sameHplus1}, $T_n$ minimizes the HM, and therefore $T_{n-1} \in B_{n-1}$.

This completes the proof.
\end{proof}

Next, we analyze the relationship between the minimizing trees with $n$ and $2n$ leaves. This is another property that minimizing trees share with gfb-trees, as the latter can be constructed from one another in the manner described in the next lemma (see Remark \ref{Rem:gfb_properties}).

\begin{Lem}
\label{Lem:Min_n_2n_all_f}
    Let $T_{n} \in \BT$ and $T_{2n} \in \mathcal{BT}^{\ast}_{2n}$ be such that $T_{2n}$ is obtained from $T_n$ by attaching a cherry to each of its leaves. Then $T_n$ (uniquely) minimizes the height metaconcept $\Phi^{\mathcal{H}}_f$ for all strictly increasing functions $f$ on $\BT$ if and only if $T_{2n}$ (uniquely) minimizes the height metaconcept $\Phi^{\mathcal{H}}_f$ for all strictly increasing functions $f$ on $\mathcal{BT}^{\ast}_{2n}$.
\end{Lem}

Note that this lemma implies a simple constructive relationship between minimizing trees with $n$ and $2n$ leaves. Specifically, every minimizing tree with $n$ leaves can be obtained by deleting all cherries from some minimizing tree with $2n$ leaves. Conversely, every minimizing tree with $2n$ leaves can be obtained by attaching a cherry to each leaf of some minimizing tree with $n$ leaves.

\begin{proof}[Proof of Lemma \ref{Lem:Min_n_2n_all_f}] 
    First, we analyze the relationship between the height sequences of two binary trees $T' \in \BT$ and $T'' \in \mathcal{BT}^{\ast}_{2n}$ when $T''$ is constructed from $T'$ by attaching a cherry to each of its leaves. By construction, $T''$ has $n$ ones in its height sequence corresponding to the $n$ parents of the attached cherries. Moreover, it contains all elements of $\mathcal{H}(T')$ increased by one, i.e.,
    \[\mathcal{H}\left(T''\right)_i = 1 \text{ for } i = 1, \ldots, n,\]
    and
    \[\mathcal{H}\left(T''\right)_i = \mathcal{H}\left(T'\right)_{i-n} +1 \text{ for } i = n+1, \ldots, 2n-1.\]

    Now, let $T_n$ and $T_{2n}$ be as described in the lemma. First, we prove that $T_{2n}$ (uniquely) minimizes the HM for all strictly increasing functions $f$ if $T_n$ (uniquely) minimizes it for all strictly increasing functions $f$. Assume $T_n$ (uniquely) minimizes the HM for all strictly increasing functions $f$, i.e., $\Phi^{\mathcal{H}}_f(T_n) \leq \Phi^{\mathcal{H}}_f(T)$ for all $T \in \BT\setminus\{T_n\}$ and all strictly increasing functions $f$ with strict inequality in case of unique minimization. By Lemma \ref{Lem:Seq_entries_Min_Max_meta_binary}, Parts 1(a) and 2(a), we have \[\mathcal{H}(T_n)_i \leq \mathcal{H}(T)_i \text{ for all } i = 1, \ldots, n-1 \text{ and all } T \in \BT\setminus\{T_n\},\]
    where, in the case of unique minimization, at least one inequality is strict.

    Seeking a contradiction, assume that $T_{2n}$ does not (uniquely) minimize the HM for all strictly increasing functions $f$ on $\mathcal{BT}^{\ast}_{2n}$. Then there exists a minimizing tree $\widetilde{T}_{2n} \in \mathcal{BT}^{\ast}_{2n}$ and some strictly increasing function $f$ such that $\Phi^{\mathcal{H}}_f\left(\widetilde{T}_{2n}\right) < \Phi^{\mathcal{H}}_f\left(T_{2n}\right)$, (or $\Phi^{\mathcal{H}}_f\left(\widetilde{T}_{2n}\right) \leq \Phi^{\mathcal{H}}_f\left(T_{2n}\right)$ in case of unique minimization). By Lemma \ref{Lem:T_Min_H_properties}, all leaves of $\widetilde{T}_{2n}$ are part of a cherry. Let $\widetilde{T}_n$ be the tree obtained from $\widetilde{T}_{2n}$ by deleting all its cherries.

    From the first part of this proof, we know that
    \[\mathcal{H}\left(T_{2n}\right)_i = \mathcal{H}\left(\widetilde{T}_{2n}\right)_i = 1 \text{ for } i = 1, \ldots, n\]
    and
    \[\mathcal{H}\left(T_{2n}\right)_i = \mathcal{H}\left(T_n\right)_{i-n} +1 \leq \mathcal{H}\left(\widetilde{T}_n\right)_{i-n} +1 = \mathcal{H}\left(\widetilde{T}_{2n}\right)_i \text{ for  } i = n+1, \ldots, 2n-1,\]
    where, in the case of unique minimization, the inequality is strict for at least one $i$.

    Since $f$ is strictly increasing, it follows that
    \[\Phi^{\mathcal{H}}_f\left(T_{2n}\right) = \sum\limits_{i=1}^{2n-1} f\left(\mathcal{H}\left(T_{2n}\right)_i\right) \leq \sum\limits_{i=1}^{2n-1} f\left(\mathcal{H}\left(\widetilde{T}_{2n}\right)_i\right) = \Phi^{\mathcal{H}}_f\left(\widetilde{T}_{2n}\right),\]
    with strict inequality in the case of unique minimization. This contradicts the assumption that $T_{2n}$ does not (uniquely) minimize the HM for all strictly increasing functions $f$, completing the proof of the if-direction.\\

    Second, we prove the only-if-direction, which proceeds in a similar way. Let $T_{2n}$ (uniquely) minimize the HM for all strictly increasing functions $f$, i.e., $\Phi^{\mathcal{H}}_f(T_{2n}) \leq \Phi^{\mathcal{H}}_f(T)$ for all $T \in \mathcal{BT}^{\ast}_{2n}\setminus\{T_{2n}\}$ and all strictly increasing $f$ with strict inequality in the case of unique minimization. By Lemma \ref{Lem:Seq_entries_Min_Max_meta_binary}, Parts 1(a) and 2(a), we have 
    \[\mathcal{H}(T_{2n})_i \leq \mathcal{H}(T)_i \text{ for all } i = 1, \ldots, 2n-1 \text{ and all } T \in \mathcal{BT}^{\ast}_{2n}\setminus\{T_{2n}\},\] where, in the case of unique minimization, at least one inequality is strict.

    Seeking a contradiction, assume that $T_{n}$ does not (uniquely) minimize the HM for all strictly increasing functions $f$ on $\BT$. Then there exists a minimizing tree $\widetilde{T}_{n} \in \BT$ and some strictly increasing function $f$ such that $\Phi^{\mathcal{H}}_f\left(\widetilde{T}_{n}\right) < \Phi^{\mathcal{H}}_f\left(T_{n}\right)$ (or $\Phi^{\mathcal{H}}_f\left(\widetilde{T}_{n}\right) \leq \Phi^{\mathcal{H}}_f\left(T_{n}\right)$ in case of unique minimization). Let $\widetilde{T}_{2n}$ be the tree obtained from $\widetilde{T}_{n}$ by attaching a cherry to each of its leaves.

    From the first part of this proof, we know that
    \[\mathcal{H}\left(T_n\right)_{i} = \mathcal{H}\left(T_{2n}\right)_{n+i} -1 \leq \mathcal{H}\left(\widetilde{T}_{2n}\right)_{n+i} -1 = \mathcal{H}\left(\widetilde{T}_n\right)_{i} \text{ for } i = 1, \ldots, n-1,\]
    where, in the case of unique minimization, at least one inequality is strict.

    Since $f$ is strictly increasing, it follows that
    \[\Phi^{\mathcal{H}}_f\left(T_{n}\right) = \sum\limits_{i=1}^{n-1} f\left(\mathcal{H}\left(T_{n}\right)_i\right) \leq \sum\limits_{i=1}^{n-1} f\left(\mathcal{H}\left(\widetilde{T}_{n}\right)_i\right) = \Phi^{\mathcal{H}}_f\left(\widetilde{T}_{n}\right),\]
    with strict inequality in the case of unique minimization. This contradicts the assumption that $T_n$ does not (uniquely) minimize the HM for all strictly increasing functions $f$, completing the proof of the only-if-direction, and thus the proof of the lemma.
\end{proof}

Having established the previous results, we are now in a position to prove Theorem \ref{Theo:Min_H_gfb_characterization}, which provides a complete characterization of the trees that minimize the HM on $\BT$.

\begin{proof}[Proof of Theorem \ref{Theo:Min_H_gfb_characterization}]
\label{Proof:Theo_Min_H_gfb_characterization}
    Let $B^{min}_n$ denote the set of binary trees that minimize the HM $\Phi^{\mathcal{H}}_f$ on $\BT$ for all (not necessarily strictly) increasing functions $f$. There are three statements to prove. We proceed by establishing them one by one, beginning with the first.
    \begin{enumerate}
        \item We first show that $T^{gfb}_n \in B^{min}_n$, i.e., the gfb-tree minimizes the HM on $\BT$ for all increasing functions $f$, and $\mathcal{H}\left(T^{gfb}_n\right)_i \leq \mathcal{H}\left(T\right)_i$ for all $i = 1, \ldots, n-1$ and $T \in \BT$. To this end, we begin by proving that the gfb-tree minimizes the HM on $\BT$ for all \textit{strictly} increasing functions. The full statement then follows from Lemma \ref{Lem:Seq_entries_Min_Max_meta_binary}, Part 1.

        For the sake of a contradiction, assume that the statement is false. Then there exists a smallest positive integer $m$ such that there is a binary tree $T^{min}_m \in \mathcal{BT}^{\ast}_m$ with $\Phi^{\mathcal{H}}_{f}\left(T^{min}_m\right) < \Phi^{\mathcal{H}}_{f}\left(T^{gfb}_m\right)$ for some strictly increasing function $f$.

        We first argue that $m > 3$ and that $m$ must be odd. For $m \leq 3$, there exists only one tree, namely the gfb-tree, so there is nothing to show. Moreover, by Remark \ref{Rem:gfb_properties}, the gfb-tree with an even number of leaves can be obtained from the gfb-tree with one leaf less by attaching a cherry to its unique leaf that is not part of a cherry. Hence, by Proposition \ref{Prop:H_Min_n_to_Min_n-1} and the minimality of $m$, $m$ must be odd.

        Let $n \coloneqq m+1$ (so that $n$ is even and $m = n-1$ is odd). By assumption and Lemma \ref{Lem:T_Min_H_properties}, tree $T^{min}_m$ has exactly one leaf that is not part of a cherry. Attaching a cherry to this leaf yields a tree $T^{min}_n$ on $n$ leaves. By Proposition \ref{Prop:H_Min_n-1_to_Min_n}, tree $T^{min}_n$ minimizes the HM on $\BT$, i.e., $\Phi^{\mathcal{H}}_{f}(T^{min}_n) \leq \Phi^{\mathcal{H}}_{f}\left(T^{gfb}_n\right)$.

        By Lemma \ref{Lem:T_Min_H_properties}, all leaves of $T^{min}_n$ are part of a cherry, since $T^{min}_n$ minimizes $\Phi^{\mathcal{H}}_{f}$. Likewise, all leaves of the gfb-tree $T^{gfb}_n$ are part of cherries. Let $T_{\frac{n}{2}} \in \mathcal{BT}^{\ast}_{\frac{n}{2}}$ be the tree obtained from $T^{min}_n$ by deleting all of its cherries. By Remark \ref{Rem:gfb_properties}, deleting all cherries of $T^{gfb}_n$ yields $T^{gfb}_{\frac{n}{2}}$. Note that $\frac{n}{2} < m$ since $m>3$.

        By the minimality of $m$, the gfb-tree $T^{gfb}_{\frac{n}{2}}$ minimizes the HM on $\mathcal{BT}^{\ast}_{\frac{n}{2}}$ for all strictly increasing functions $f$. Hence, by Lemma \ref{Lem:Min_n_2n_all_f}, the gfb-tree $T^{gfb}_n$ minimizes the HM on $\BT$ for all strictly increasing functions $f$. Finally, by Lemma \ref{Prop:H_Min_n_to_Min_n-1}, the gfb-tree $T^{gfb}_{n-1} = T^{gfb}_{m}$ minimizes the HM on $\mathcal{BT}^{\ast}_{n-1}$ for all strictly increasing functions $f$. This contradicts the assumption that $\Phi^{\mathcal{H}}_{f}\left(T^{min}_m\right) < \Phi^{\mathcal{H}}_{f}\left(T^{gfb}_{m}\right)$ for some strictly increasing function $f$. Therefore, the gfb-tree minimizes the HM for all strictly increasing functions $f$.

        Finally, by Lemma \ref{Lem:Seq_entries_Min_Max_meta_binary}, Part 1, the gfb-tree also minimizes the HM for all increasing functions $f$, and the statement regarding the height sequence follows. This completes this part of the proof.

        \item Let $T^{min}$ be a binary tree that minimizes the HM on $\BT$ for some strictly increasing function $f$. Note that $T^{min} \in B^{min}_n$ is possible. To prove Parts 2.i) and 2.ii) of Theorem \ref{Theo:Min_H_gfb_characterization}, we show that $T^{min}$ has the same height sequence as the gfb-tree.

        By the first part of the proof, the gfb-tree minimizes the HM on $\BT$ for all (not necessarily strictly) increasing functions, i.e., $T^{gfb}_n \in B^{min}_n$, and it satisfies
        \[\mathcal{H}\left(T^{gfb}_n\right)_i \leq \mathcal{H}\left(T\right)_i \text{ for all } i = 1, \ldots, n-1, \text{ and all } T \in \BT.\] 
        To prove the claim, we show that $T^{min}$ has the same height sequence as the gfb-tree. Consequently, it also minimizes the HM on $\BT$ for all (strictly) increasing functions, i.e., $T^{min} \in B^{min}_n$, and all trees in $B^{min}_n$ have the same height sequence as the gfb-tree.

        By assumption and the first part of the proof, we have \[\Phi^{\mathcal{H}}_f(T^{min}) = \Phi^{\mathcal{H}}_f\left(T^{gfb}_n\right)\] for some strictly increasing function $f$, and 
        \[\mathcal{H}\left(T^{gfb}_n\right)_i \leq \mathcal{H}(T^{min})_i \text{ for all } i = 1, \ldots, n-1.\] 
        Hence, $T^{min}$ must have exactly the same height sequence as the gfb-tree, since $f$ is strictly increasing; otherwise we would have
        \[\Phi^{\mathcal{H}}_f\left(T^{gfb}_n\right) = \sum\limits_{i=1}^{n-1} f\left(\mathcal{H}\left(T^{gfb}_n\right)_i\right) < \sum\limits_{i=1}^{n-1} f\left(\mathcal{H}\left(T^{min}\right)_i\right) = \Phi^{\mathcal{H}}_f\left(T^{min}\right),\]
        a contradiction.

        In particular,
        \[h(T^{min}) = \mathcal{H}(T^{min})_{n-1} = \mathcal{H}\left(T^{gfb}_n\right)_{n-1} = h\left(T^{gfb}_n\right) = \left\lceil\log_2(n)\right\rceil,\]
        which is the minimal possible height for any $T \in \BT$ (see Remark \ref{Rem:gfb_properties}). This completes this part of the proof.

        \item Finally, we show that all pending subtrees of $T^{min} \in B^{min}_n$ are minimizing trees as well. We prove this statement for maximal pending subtrees; the general case then follows recursively.

        Let $T^{min} = (T_1, T_2) \in B^{min}_n$ be a minimizing tree of the HM on $\BT$ for all increasing functions $f$, where $T_1$ and $T_2$ are its maximal pending subtrees with $n_1$ and $n_2$ leaves, respectively. 

        Without loss of generality, assume that $T_1 \notin B^{min}_{n_1}$, i.e., $T_1$ does not minimize the HM on $\mathcal{BT}^{\ast}_{n_1}$. Consider tree $\widetilde{T} = \left(T^{gfb}_{n_1},T_2\right)$. By Remark \ref{Rem:gfb_properties}, we have $h(T_1) \geq h\left(T^{gfb}_{n_1}\right)$, and thus $h\left(T^{min}\right) \geq h(\widetilde{T})$. We now show that $\widetilde{T}$ attains a strictly smaller value of the HM than $T^{min}$, contradicting the minimality of $T^{min}$:
        \begin{align*}
            \Phi^{\mathcal{H}}_f\left(T^{min}\right) - \Phi^{\mathcal{H}}_f(\widetilde{T}) &= f\left(h\left(T^{min}\right)\right) + \Phi^{\mathcal{H}}_f(T_1) + \Phi^{\mathcal{H}}_f(T_2) - \left(f(h(\widetilde{T})) + \Phi^{\mathcal{H}}_f\left(T^{gfb}_{n_1}\right) + \Phi^{\mathcal{H}}_f(T_2)\right)\\
            &\geq \Phi^{\mathcal{H}}_f(T_1) - \Phi^{\mathcal{H}}_f\left(T^{gfb}_{n_1}\right) > 0.
        \end{align*}
        The first inequality holds because $h\left(T^{min}\right) \geq h(\widetilde{T})$, and $f$ is increasing. The second inequality follows from the assumption that $T_1$ does not minimize the HM, together with the first part of the proof, which states that the gfb-tree minimizes the HM.

        This contradicts the minimality of $T^{min}$ and shows that every maximal pending subtree of a minimizing tree must itself be minimizing. By recursion, it follows that all pending subtrees of binary minimizing trees are minimizing trees.
    \end{enumerate}
    This completes the proof.
\end{proof}

Next, we state a direct consequence of Theorem \ref{Theo:Min_H_gfb_characterization}, which will be used frequently in subsequent proofs.

\begin{Cor}
\label{Cor:Min_n_2n}
    Let $f$ be a strictly increasing function. Let $T_{n} \in \BT$ and $T_{2n} \in \mathcal{BT}^{\ast}_{2n}$ be such that $T_{2n}$ is obtained from $T_n$ by attaching a cherry to each of its leaves. Then $T_n$ (uniquely) minimizes the height metaconcept $\Phi^{\mathcal{H}}_f$ on $\BT$ if and only if $T_{2n}$ (uniquely) minimizes the height metaconcept $\Phi^{\mathcal{H}}_f$ on $\mathcal{BT}^{\ast}_{2n}$.
\end{Cor}
\begin{proof}
    The proof follows directly from Lemma \ref{Lem:Min_n_2n_all_f} and Theorem \ref{Theo:Min_H_gfb_characterization}, Part 2(ii).
\end{proof}

We now know several properties of the minimizing trees, in particular that all share the same height sequence as the gfb-tree. We next turn our attention to the number of minimizing trees.

\subsubsection{The number of trees minimizing the HM}

In Proposition \ref{Prop:star_Min_on_T}, we have seen that the star tree is the unique tree minimizing the HM on $\T$ for all $1$-positive functions $f$. On the other hand, we proved in Theorem \ref{Theo:Min_H_gfb_characterization} that the gfb-tree minimizes the HM for all increasing functions on $\BT$. We now describe cases in which the gfb-tree is the unique minimizer. For $n = 2^h$ the HM has to be uniquely minimized by the gfb-tree (which, in this case, coincides with the fb-tree of height $h$) in order to be a (binary) imbalance index.

\begin{Prop}
\label{Prop:gfb_unique_Min}
    Let $f$ be strictly increasing. Then the gfb-tree $T^{gfb}_n$ is the unique tree minimizing the height metaconcept $\Phi^{\mathcal{H}}_f$ on $\BT$ if and only if $n = 2^h - 2^i$ for integers $h > i \geq 0$.

    Moreover, in the case of unique minimization of the gfb-tree, we have for all binary trees $T \in \BT \setminus \left\{T^{gfb}_{n}\right\}$ that
    \[\mathcal{H}\left(T^{gfb}_n\right)_i \leq \mathcal{H}\left(T\right)_i \text{ for all } i \in \left\{1, \ldots, n-1\right\}\]
    and
    \[\mathcal{H}\left(T^{gfb}_n\right)_i < \mathcal{H}\left(T\right)_i \text{ for at least one } i \in \left\{1, \ldots, n-1\right\}.\]
\end{Prop}

Note that by choosing $i = h-1$, the number of leaves $n$ can attain any power of two as then $n = 2^h - 2^i = 2^{h-1}$. 

Moreover, the only \textit{odd} leaf numbers for which the minimizing tree is unique are those with $i=0$, i.e., $n = 2^h-1$ for some positive integer $h$.

\begin{proof}
    Let $f$ be a strictly increasing function. We first prove the if-and-only-if statement. By Theorem \ref{Theo:Min_H_gfb_characterization}, Part 1, we know that the gfb-tree minimizes the HM for all $n$. Hence, it suffices to show that there is a unique tree minimizing the HM for leaf numbers of the form $n = 2^h - 2^i$ for arbitrary integers $h > i \geq 0$, and that there are at least two distinct trees minimizing the HM if $n \neq 2^h - 2^i$ for all $h > i \geq 0$.

    First, let $n = 2^h - 2^i$ for arbitrary integers $h > i \geq 0$. We have to show that there is a unique tree minimizing the HM. 
    Our proof proceeds by double induction on the pairs $(h,i)$ with $h > i \geq 0$.

    For the base case, consider all pairs $(h,0)$ with $h \geq 1$, i.e., $i = 0$ and $n = 2^h - 1$. If $h = 1$, we have $n = 1$, and the statement clearly holds (as there exists only one tree with $n = 1$ leaves). Thus, assume in the following that $h \geq 2$. We first show that there is a unique tree minimizing the HM on $n+1 = 2^h$ leaves. By Theorem \ref{Theo:Min_H_gfb_characterization}, Part 2(i), any minimizing tree with $n+1 = 2^h$ leaves must have height $h$. The fb-tree $T^{fb}_{h}$ is the only tree of height $h$, and it coincides with the gfb-tree $T^{gfb}_{2^h}$. Therefore, the minimizer for $n+1 = 2^h$ leaves is unique.

    Furthermore, due to the symmetry of $T^{fb}_h$, deleting any cherry produces the same tree on $n = 2^h -1$ leaves. By Proposition \ref{Prop:H_Min_n_to_Min_n-1}, this resulting tree is the unique HM-minimizing tree for $n=2^h-1$. Hence, the base case holds for all pairs $(h,0)$.

    For the induction step, assume that the statement holds for the pair $(h-1,i-1)$ with $0 \leq i-1 < h-1$. Now consider the pair $(h,i)$ with $h > i > 0$ (note that the latter inequality ensures that we are no longer in the base case of the induction), so that $n = 2^h - 2^i$. By the induction hypothesis, there exists a unique tree, say $T_1$, that minimizes the HM for $\frac{n}{2} = 2^{h-1} - 2^{i-1}$. Let $T_2$ be the tree obtained from $T_1$ by attaching a cherry to each of its leaves. Then, by Corollary \ref{Cor:Min_n_2n}, $T_2$ is the unique tree minimizing the HM for $n = 2^h - 2^i$, and hence for the pair $(h,i)$. This completes the double induction.

    \medskip

    It remains to show that there exist at least two distinct trees minimizing the HM if $n \neq 2^h - 2^i$ for all $h > i \geq 0$. We consider two cases, depending on whether $n$ is even or odd.

    First, suppose $n$ is odd and $n \neq 2^h - 2^i$ for all $h > i \geq 0$. In particular, $n \neq 2^h -1 = 2^h - 2^0$, and thus $n \geq 5$. By Proposition \ref{Prop:H_Min_n_to_Min_n-1}, any minimizing tree with $n$ leaves can be obtained from a minimizing tree with $n+1$ leaves by deleting an arbitrary cherry. Let $T_{n+1}$ be a minimizing tree with $n+1$ leaves. Since $n+1 \neq 2^h$, tree $T_{n+1}$ must contain leaves of different depths. In fact, by Lemma \ref{Lem:T_Min_H_properties}, $T_{n+1}$ contains two cherries at different depths. Deleting either of these cherries yields two distinct trees with $n$ leaves. Therefore, at least two distinct trees with $n$ leaves minimize the HM.

    Second, let $n$ be even and $n \neq 2^h - 2^i$ for all $h > i > 0$. Note that the case $i = 0$ cannot occur here, since it would imply that $n$ is odd. Moreover, we must have $n > 4$, because $2 = 2^2 - 2^1$ and $4 = 2^3 - 2^2$. Now, let $n$ be the smallest such leaf number for which there is a unique tree minimizing the HM. We now distinguish two cases, depending on whether $\frac{n}{2}$ is even or odd. As we shall see, both cases will lead to a contradiction, implying that no such $n$ exists.

    \begin{itemize}
        \item Suppose first that $\frac{n}{2}$ is even. As, by assumption, $n$ cannot be expressed as a difference of two powers of two, we also have $\frac{n}{2} \neq 2^h - 2^i$ for all $h > i > 0$. By the minimality of $n$, there must be at least two distinct trees minimizing the HM for $\frac{n}{2}$ leaves. By Corollary \ref{Cor:Min_n_2n}, this implies that there are also at least two distinct trees minimizing the HM for $n$ leaves. This contradicts the assumed uniqueness of the minimizer for $n$.
        \item Now suppose that $\frac{n}{2}$ is odd. Again, as $n$ cannot be expressed as a difference of two powers of two, we also have $\frac{n}{2} \neq 2^h - 2^i$ for all $h > i \geq 0$. In particular, $\frac{n}{2} \neq 2^h - 1 = 2^h - 2^0$ for all $h$, so the case $i = 0$ is excluded. By the previous part of the proof (the case for odd $n$), there must again be at least two distinct trees minimizing the HM for $\frac{n}{2}$ leaves. Applying Corollary \ref{Cor:Min_n_2n} once more, we conclude that there are at least two distinct trees minimizing the HM for $n$ leaves as well. This contradiction completes the proof of the if-and-only-if statement.
    \end{itemize}

    Finally, since $f$ was arbitrary, the corresponding statement for the entries of the height sequence follows directly from the preceding arguments together with Lemma \ref{Lem:Seq_entries_Min_Max_meta_binary}, Parts 1(a) and 2(a).
\end{proof}

Besides the gfb-tree, there exist additional families of trees that also minimize the HM on $\BT$. These families of trees share certain structural properties, which are stated in the following proposition. Note that, by Remark \ref{Rem:gfb_properties}, the gfb-tree itself satisfies these properties as well.

\begin{Prop}\label{Prop:H_Min_T_not_gfb}
    Let $\{T_n\}_{n\in\N}$ be a family of binary trees satisfying the following properties:
    \begin{enumerate}
        \item For every $n \in \mathbb{N}$, attaching a cherry to each leaf of $T_n$ yields $T_{2n}$.
        \item If $n$ is even, then every leaf of $T_n$ is part of a cherry.
        \item If $n$ is odd, then $T_n$ has exactly one leaf that is not part of a cherry, and attaching a cherry to this leaf yields $T_{n+1}$.
    \end{enumerate}
    Then $T_n$ has the same height sequence as the gfb-tree $T^{gfb}_n$ and therefore minimizes the height metaconcept $\Phi^{\mathcal{H}}_f$ for all (not necessarily strictly) increasing functions $f$ on $\BT$.
\end{Prop}
\begin{proof}
    Note that having the same height sequence as the gfb-tree is equivalent to minimizing the HM (see Theorem \ref{Theo:Min_H_gfb_characterization}, Part 2(i)). First, let $f$ be a strictly increasing function. We show that $T_n \in \{T_n\}_{n\in\N}$ minimizes the HM for all $n$.

    Let $n$ be the minimum number of leaves such that $T_n$ does \emph{not} minimize the HM. For $n \leq 3$, only one tree exists, which coincides with the gfb-tree. Hence, this tree must belong to the family $\{T_n\}_{n\in\N}$. Now, let $n \geq 4$. Next, we consider two cases, depending on whether $n$ is even or odd.

    First, assume that $n$ is even. By the minimality of $n$, the tree $T_{\frac{n}{2}}$ minimizes the HM. Hence, by Corollary \ref{Cor:Min_n_2n}, it follows that $T_n$ also minimizes the HM, which contradicts the choice of $n$.

    Second, assume that $n$ is odd. In this case, $T_{n+1}$ can be obtained from $T_n$ by attaching a cherry to its unique leaf that is not part of a cherry. Since $n \geq 4$, the same argument as in the even case shows that $T_{n+1}$ minimizes the HM. Then, by Proposition \ref{Prop:H_Min_n_to_Min_n-1}, it follows that $T_n$ also minimizes the HM, again contradicting the choice of $n$.

    Hence, $T_n \in \{T_n\}_{n\in\N}$ minimizes the HM for some strictly increasing function $f$ for all $n$. The statement now follows from Theorem \ref{Theo:Min_H_gfb_characterization}, Part 2.

    This completes the proof.
\end{proof}

Next, we show that, in addition to the gfb-tree, the echelon tree also minimizes the HM.

\begin{Prop}
\label{Prop:H_echelon_Min}
    The echelon tree has the same height sequence as the gfb-tree $T_n^{gfb}$ and thus minimizes the height metaconcept $\Phi^{\mathcal{H}}_f$ for all (not necessarily strictly) increasing functions $f$ on $\BT$.
\end{Prop}

Before proving this proposition, we show that the echelon tree satisfies the properties stated in Proposition \ref{Prop:H_Min_T_not_gfb}.

\begin{Lem}
\label{Lem:echelon_properties}
    Let $T^{be}_n$ be the echelon tree with $n$ leaves. Attaching a cherry to each leaf of $T^{be}_n$ yields $T^{be}_{2n}$. Moreover, if $n$ is even, every leaf of $T^{be}_n$ belongs to a cherry. If $n$ is odd, $T^{be}_n$ has a unique leaf that is not part of a cherry, and attaching a cherry to this leaf yields $T^{be}_{n+1}$.
\end{Lem}
\begin{proof}
    For the proof, recall that the echelon tree can be constructed from a caterpillar on $w(n)$ leaves, where $w(n)$ is the binary weight of $n$, i.e., the number of non-zero summands in its binary expansion. In the construction, each leaf of the caterpillar is replaced by an fb-tree whose height is determined by the binary expansion of $n$ (see Section~\ref{Sec:Prelim} for details).

    First, we show that at most one leaf of $T^{be}_n$ is not part of a cherry. This follows immediately from the fact that the summand $2^0 = 1$ occurs in the binary expansion of $n$ if and only if $n$ is odd; thus, there is exactly one fb-tree of height $0$ if $n$ is odd, and hence at most one leaf not belonging to a cherry for all $n$.

    Second, we show that attaching a cherry to each leaf of $T^{be}_n$ yields $T^{be}_{2n}$. Observe that $n$ and $2n$ have the same binary weight, since the binary expansion of $2n$ is obtained from that of $n$ by multiplying every summand by 2. Consequently, both $T^{be}_n$ and $T^{be}_{2n}$ are constructed from the same caterpillar. The only difference is that, for $T^{be}_{2n}$, each attached fb-tree has twice the size of the corresponding fb-tree used for $T^{be}_n$. This proves the second claim.

    Finally, we prove that if $n$ is odd, attaching a cherry to the unique leaf of the echelon tree that is not part of a cherry yields $T^{be}_{n+1}$. Suppose, for contradiction, that this statement fails, and let $n$ be the smallest number of leaves for which it does not hold. Note that this implies that $n>3$, because for $n=3$, there is only one tree, and attaching a cherry to its singleton leaf yields $T_2^{fb}$, i.e., the fb-tree on 4 leaves, which is obviously an echelon tree. Thus, we now assume $n\geq 5$. We can now express $n$ as follows: $n = 2^h + p$, where $h = \lfloor\log_2(n)\rfloor$ and $p \in \{1,3, \ldots, 2^{h}-1\}$ as $n$ is odd. Then, by definition, $T^{be}_n = \left(T^{fb}_h,T^{be}_p\right)$ and $T^{be}_{n+1} = \left(T^{fb}_h,T^{be}_{p+1}\right)$. Thus, it remains to show that attaching a cherry to $T^{be}_n$ yields $\left(T^{fb}_h,T^{be}_{p+1}\right)$. This follows because the unique leaf in $T^{be}_n$ that is not part of a cherry lies in the pending subtree $T^{be}_p$. By the minimality of $n$, $T^{be}_{p+1}$ can be obtained from $T^{be}_p$ by attaching a cherry to this unique leaf. This contradiction completes the proof of the claim.
\end{proof}

We are now in a position to prove Proposition \ref{Prop:H_echelon_Min}.

\begin{proof}[Proof of Proposition \ref{Prop:H_echelon_Min}]
    The statement for the echelon tree follows directly from Proposition \ref{Prop:H_Min_T_not_gfb} and Lemma \ref{Lem:echelon_properties}.
\end{proof}

Next, we show that all pending subtrees of a minimizing tree can be replaced by another minimizing tree of the same size of the pending subtree in order to obtain another minimizing tree.

\begin{Prop}
\label{Prop:min_H_replacing_subtree}
    Let $T^{min}_n$ be a tree containing a pending subtree $T_v$ of size $n_v$ that minimizes the height metaconcept for all increasing functions $f$ on $\BT$. Let $T^{min}_v$ be another tree with $n_v$ leaves that minimizes the height metaconcept for all increasing functions $f$ on $\mathcal{BT}^{\ast}_{n_v}$. Let $T'$ be the tree obtained from $T^{min}_n$ by replacing the subtree $T_v$ with $T^{min}_v$. Then $T'$ also minimizes the height metaconcept for all increasing functions $f$ on $\BT$.
\end{Prop}
\begin{proof}
    We prove the statement for $T_v$ being a maximal pending subtree of $T^{min}_n$; the case of an arbitrary pending subtree then follows recursively.

    Let $T^{min}_n = (T_1,T_2)$ and $T^{min}_1$ be trees with $n$ and $n_1$ leaves, respectively, that minimize the height metaconcept for all increasing functions $f$ on $\BT$ and $\mathcal{BT}^{\ast}_{n_1}$, respectively. Let $T_1$ have $n_1$ leaves, and let $T' = (T^{min}_1,T_2)$. We show that $T'$ also minimizes the HM by proving that $T'$ and $T^{min}_n$ have the same height sequence.

    First, note that $T^{min}_n$ and $T'$ share $T_2$ as a maximal pending subtree. Hence, the entries in their height sequences corresponding to the inner vertices of $T_2$ are identical. Similarly, the inner vertices of $T_1$ and $T^{min}_1$ also have identical entries in their height sequences, because both $T_1$ and $T^{min}_1$ minimize the HM on $\mathcal{BT}^{\ast}_{n_1}$ (by assumption and, in case of $T_1$ as a subtree of $T_n^{min}$, by Part 3 of Theorem \ref{Theo:Min_H_gfb_characterization}). 
    
    Finally, since the heights of the two trees are determined by the maximal heights of their subtrees, $T^{min}_n$ and $T'$ also have the same overall height. Therefore, their height sequences coincide, which completes the proof.
\end{proof}

So far, we have thoroughly analyzed the trees that maximize and minimize the HM on $\BT$ and $\T$. Next, we use these results to determine the maximum and minimum values of the HM.

\subsubsection{Maximum and minimum values}

First, we calculate the maximum value of the HM on $\BT$ and $\T$.

\begin{Cor}
\label{Cor:Meta_Max_val}
Let $f$ be a function. Then $\Phi^{\mathcal{H}}_f\left(T^{cat}_n\right) = \sum\limits_{i=1}^{n-1} f(i)$ is the maximum value of the height metaconcept $\Phi^{\mathcal{H}}_f$
\begin{itemize}
    \item on $\T$, if $f$ is increasing and $1$-positive, and
    \item on $\BT$ if $f$ is increasing.
\end{itemize}
\end{Cor}
\begin{proof}
By Theorem \ref{Theo:cat_Max_H}, Parts 1 and 3, it suffices to show that $\Phi^{\mathcal{H}}_f\left(T^{cat}_n\right) = \sum\limits_{i=1}^{n-1} f(i)$. The parent of the unique cherry in the caterpillar has height value $1$, while its root has height value $n-1$. In particular, the caterpillar contains exactly one pending subtree of each height from 
$1$ to $n-1$. Thus, $\Phi^{\mathcal{H}}_f\left(T^{cat}_n\right) = \sum\limits_{i=1}^{n-1} f(i)$, which completes the proof.
\end{proof}

To determine the minimum value of the HM, we first establish the following proposition concerning the number of pending subtrees of the gfb-tree with a given height.

\begin{Prop}
\label{Prop:gfb_number_subtrees_certain_height}
Let $n \geq 2$ with $n = 2^{h_n} + p_n$ with $h_n = \lfloor \log_2(n) \rfloor$ and $0 \leq p_n < 2^{h_n}$. Let $T^{gfb}_n$ be the gfb-tree with $n$ leaves, and let $A_n(i)$ denote the number of pending subtrees of height $i$ of the gfb-tree.

\begin{enumerate}
    \item If $i > \lceil \log_2(n) \rceil$, then $A_n(i)=0$.
    \item If $i \leq \lceil \log_2(n) \rceil$, then:
        \begin{enumerate}[(i)]
            \item If $p_n = 0$, i.e., $n = 2^{h_n}$, then $A_n(i) = 2^{h_n-i}$ for $i=1, \ldots, h_n$.
            \item If $0 < p_n < 2^{h_n}$, then for $i=1, \ldots, h_n+1$, we have \[A_n(i) = 2^{h_n+1-i} - \left\lceil \frac{2^{h_n}-p_n-2^{i-1}+1}{2^i} \right\rceil.\]
        \end{enumerate}
\end{enumerate} 
\end{Prop}

\begin{proof}\leavevmode
\begin{enumerate}
    \item First, assume that $i > \lceil \log_2(n) \rceil$. Recall that the height of the gfb-tree with $n$ leaves is $\lceil \log_2(n) \rceil$. Consequently, the gfb-tree does not contain any pending subtrees of height strictly greater than $\lceil \log_2(n) \rceil$, which immediately implies $A_n(i)=0$ for all $i > \left \lceil \log_2(n) \right \rceil$.

    \item Now assume that $1 \leq i \leq \lceil \log_2(n) \rceil$.

    \begin{enumerate}[(i)]
        \item  If $p_n = 0$, then $T^{gfb}_n$ is the fully balanced tree of height $h_n$ and, by definition, contains $2^{h_n-i}$ pending subtrees of height $i$. This completes the proof for $p_n = 0$.

        \item Now assume that $0 < p_n < 2^{h_n}$. We prove this statement by induction on $n$. For $n = 3$, the gfb-tree $T^{gfb}_3$ contains one pending subtree of height $i=1$ ($T^{gfb}_3$'s cherry) and one pending subtree of height $i=2$ ($T^{gfb}_3$ itself). Moreover, $h_3 = \lfloor \log_2(3) \rfloor = 1$ and $p_3 = 1$, and thus
        \begin{align*}
            A_3(1) &= 2^{1+1-1} - \left\lceil \frac{2^1-1-2^{1-1}+1}{2^1} \right\rceil
            = 2 - \left\lceil \frac{1}{2} \right\rceil = 1,\\
            A_3(2) &= 2^{1+1-2} - \left\lceil \frac{2^1-1-2^{2-1}+1}{2^2} \right\rceil
            = 1 - \left \lceil \frac{0}{4} \right\rceil = 1,
        \end{align*}
        which establishes the base case.

        Now assume that the statement is true for all gfb-trees with $n' \leq n-1$ leaves and $0 < {p_{n'}} < 2^{h_{n'}}$, where $n' = 2^{h_{n'}}+p_{n'}$ and $h_{n'} = \lfloor \log_2(n') \rfloor$, and consider the gfb-tree $T^{gfb}_n = (T_1,T_2)$ with $n$ leaves. Express $n$ again as $n = 2^{h_n} + p_n$ with $h_n = \lfloor \log_2(n) \rfloor$ and $0 < p_n < 2^{h_n}$.
            \begin{itemize}
                \item First, consider $i = h_n+1$. Notice that $T^{gfb}_n$ contains precisely one pending subtree of height $h_n+1$ ($T^{gfb}_n$ itself) and indeed
                \begin{align} \label{eq:an-i-hn+1}
                A_n(h_n+1) &= 2^{h_n+1-(h_n+1)} - \left\lceil \frac{2^{h_n} - p_n - 2^{(h_n+1)-1}+1}{2^{h_n+1}} \right\rceil = 1 - \left\lceil \frac{1 - p_n}{2^{h_n+1}} \right\rceil = 1-0 = 1.
                \end{align}

                Note that the $0$ in the previous line stems from the fact that $1\leq p_n<2^{h_n}$, which implies that the numerator is at least $1-(2^{h_n}-1)=-2^{h_n}+2>-2^{h_n+1}$ and at most $1-1=0$.

            \item Now, let $1 \leq i \leq h_n$. To complete the proof, we use Proposition~\ref{Prop:gfb_Ta_or_Tb_fb} and distinguish two cases based on the value of $p_n$:
            \begin{enumerate}[(a)]
                \item If $0 < p_n < 2^{h_n-1}$, then $T_1$ is the gfb-tree on $n_1 = 2^{h_n-1}+p_n$ leaves and $T_2$ is the fb-tree of height $h_n-1$. Thus, $h_{n_1} = h_n-1$ and $p_{n_1} = p_n < 2^{h_n-1} = 2^{h_{n_1}}$. By the inductive hypothesis and Part 2(i),
            \begin{itemize}
                \item $T_1$ contains $A_{n_1}(i) = 2^{h_n-i} - \left\lceil \frac{2^{h_n-1}-p_n - 2^{i-1}+1}{2^i} \right\rceil$ pending subtrees of height $i$, $1 \leq i \leq h_n$, and
                \item $T_2$ contains $A_{n_2}(i) = 2^{h_n-1-i}$ pending subtrees of height $i$, $1 \leq i \leq h_n-1$.
            \end{itemize}
            Now notice that as $p_n \neq 0$, the number of pending subtrees of a fixed height $i \leq h_n$ in $T^{gfb}_n$ is the sum of the number of pending subtrees of height $i$ in $T_1$ and $T_2$, respectively. Thus, we need to show that for $1 \leq i \leq h_n$, we have $A_{n_1}(i) + A_{n_2}(i) = A_n(i)$.
                \begin{itemize}
                \item If $i = h_n$, then $A_{n_2}(h_n) = 0$ (since $h_n > \lceil \log_2 n_2 \rceil = h_n-1$, which implies $A_{n_2}(h_n) = 0$). 
                Moreover, since $p_n > 0$, we have $h(T_{1}) = h_n$, and thus \[A_{n_1}(h_n) = 2^{h_n-h_n} - \left\lceil \frac{2^{h_n-1}-p_{n} - 2^{h_n-1}+1}{2^{h_n}} \right\rceil = 1 -  \left\lceil \frac{1-p_n}{2^{h_n}} \right\rceil = 1 - 0 = 1,
                \]
                where the second to last equality follows analogously to Eq.~\eqref{eq:an-i-hn+1}.
                
                On the other hand, \[A_n(h_n) = 2^{(h_n+1)-h_n} - \left\lceil \frac{2^{h_n} - p_n - 2^{h_n-1}+1}{2^{h_n}} \right\rceil = 2 - \left\lceil \frac{2^{h_n-1}-p_n+1}{2^{h_n}} \right\rceil = 2 -1 = 1,\] where the second to last equality follows from the fact that $0 < p_n \leq 2^{h_n-1}$. In particular, $A_{n_1}(h_n) + A_{n_2}(h_n) = 1 + 0 = 1 = A_n(h_n)$.
                \item If $i \leq h_n-1$, we have
                \begin{align*}
                A_{n_1}(i) + A_{n_2}(i) &= \left(2^{h_n-i} - \left\lceil \frac{2^{h_n-1}-p_n - 2^{i-1}+1}{2^i} \right\rceil\right) + 2^{h_n-i-1}\\
                &= 3 \cdot 2^{h_n-i-1} - \left\lceil \frac{2^{h_n-1}-p_n - 2^{i-1}+1}{2^i} \right\rceil\\
                &= 3 \cdot 2^{h_n-i-1} - \left\lceil \underbrace{2^{h_n-1-i}}_{\in \N_{\geq 1}} + \frac{1-p_n-2^{i-1}}{2^i} \right\rceil\\
                &= 3 \cdot 2^{h_n-i-1} - 2^{h_n-i-1} - \left\lceil \frac{1-p_n-2^{i-1}}{2^i} \right\rceil\\
                &= 2^{h_n-i} - \left\lceil \frac{1-p_n-2^{i-1}}{2^i} \right\rceil\\
                &= 2^{h_n+1-i} - 2^{h_n-i} - \left\lceil \frac{1-p_n-2^{i-1}}{2^i} \right\rceil\\
                &= 2^{h_n+1-i} - \left\lceil \underbrace{2^{h_n-i}}_{\in \N_{\geq 1}} + \frac{1-p_n-2^{i-1}}{2^i} \right\rceil\\
                &= 2^{h_n+1-i} - \left\lceil \frac{2^{h_n} - p_n - 2^{i-1}+1}{2^i} \right\rceil\\
                &= A_n(i).
                \end{align*}
                 Therefore, $A_{n_1}(i) + A_{n_2}(i) = A_n(i)$.
            \end{itemize}

    \item If $p_n = 2^{h_n-1}$, then $T_1$ is the fb-tree of height $h_n$ (i.e., $n_1 = 2^{h_n})$ and $T_2$ is the fb-tree of height $h_n-1$ (i.e., $n_2 = 2^{h_n-1})$. In this case, $p_{n_1} = p_{n_2}=0$. By Part 2(i),
    \begin{itemize}
        \item $T_1$ contains $A_{n_1}(i)=2^{h_n-i}$ pending subtrees of height $i$, $1 \leq i \leq h_n$, and 
        \item $T_2$ contains $A_{n_2}(i)=2^{h_n-i-1}$ pending subtrees of height $i$, $1 \leq i \leq h_n-1$.
    \end{itemize}
    Again, we need to show that $A_n(i) = A_{n_1}(i) + A_{n_2}(i)$.
    Similar to the last case, for $i=h_n$, we have $A_{n_1}(h_n) + A_{n_2}(h_n) = 1 + 0 = 1 = A_n(h_n)$.
    For $i \leq h_n-1$, on the other hand,
    \begin{align*}
        A_{n_1}(i) + A_{n_2}(i) 
        &= 2^{h_n-i} + 2^{h_n-i-1} 
        = 3 \cdot 2^{h_n-i-1},
    \end{align*}
    and
    \begin{align*}
        A_n(i)
        &= 2^{h_n+1-i} - \left\lceil \frac{2^{h_n}-2^{h_n-1}-2^{i-1}+1}{2^i} \right\rceil \\
        &= 2^{h_n+1-i} - \left\lceil \frac{2^{h_n-1}-2^{i-1}+1}{2^i} \right\rceil \\
        &= 2^{h_n+1-i} - \left\lceil 2^{h_n-1-i}-\frac{1}{2} + \frac{1}{2^i} \right\rceil \\
        &= 2^{h_n+1-i} - 2^{h_n-1-i} + \underbrace{\left\lceil -\frac{1}{2} + \frac{1}{2^i}\right\rceil}_{=0} \\
        &=  3 \cdot 2^{h_n-i-1}.
    \end{align*}
    Thus, $A_n(i) = A_{n_1}(i) + A_{n_2}(i)$ as required.

    \item If $2^{h_n-1} < p_n < 2^{h_n}$, then $T_1$ is the fb-tree of height $h_n$ and $T_2$ is the gfb-tree on $n_2 = p_n = 2^{h_n-1} + p_{n_2}$ leaves. Thus, $h_{n_2} = h_n-1$ and $p_{n_2} = p_n - 2^{h_n-1} <2^{h_n-1}=2^{h_{n_2}}$. By Part 2(i) and the inductive hypothesis,
    \begin{itemize}
        \item $T_1$ contains $A_{n_1}(i) = 2^{h_n-i}$ pending subtrees of height $i$, $1 \leq i \leq h_n$, and
        \item $T_2$ contains $A_{n_2}(i) = 2^{h_n-i} - \left\lceil \frac{2^{h_n-1} - (p_n-2^{h_n-1}) - 2^{i-1}+1}{2^i} \right\rceil = 2^{h_n-i} - \left\lceil \frac{2^{h_n} - p_n - 2^{i-1}+1}{2^i} \right\rceil$ pending subtrees of height $i$, $1 \leq i \leq h_n$.
    \end{itemize}
    The proof now proceeds as in the previous cases by showing that for $1 \leq i \leq h_n$, we have $A_{n_1}(i) + A_{n_2}(i) = A_n(i)$. Therefore, consider
    \begin{align*}
        A_{n_1}(i) + A_{n_2}(i) &= 2^{h_n-i} + \left(2^{h_n-i} - \left\lceil \frac{2^{h_n} - p_n - 2^{i-1}+1}{2^i} \right\rceil\right)\\
        &= 2^{h_n+1-i} - \left\lceil \frac{2^{h_n} - p_n - 2^{i-1}+1}{2^i} \right\rceil = A_n(i),
    \end{align*}
    which completes the proof.
    \end{enumerate} 
    \end{itemize}
    \end{enumerate}
\end{enumerate}
\end{proof}

This leads us to the minimum value of the HM on $\BT$.

\begin{Cor}
\label{Cor:Meta_Min_val}
Let $f$ be an increasing function and let $n \geq 2$ with $n = 2^{h_n} + p_n$ with $h_n = \lfloor \log_2(n) \rfloor$ and $0 \leq p_n < 2^{h_n}$. Then, the minimum value of the height metaconcept $\Phi^{\mathcal{H}}_f$ on $\BT$ is
\begin{enumerate}[(i)]
    \item $\Phi^{\mathcal{H}}_f\left(T^{gfb}_n\right) = \sum\limits_{i = 1}^{h_n} 2^{h_n-i} \cdot f(i)$, if $p_n = 0$, i.e., $n = 2^{h_n}$, and
    \item $\Phi^{\mathcal{H}}_f\left(T^{gfb}_n\right) = \sum\limits_{i = 1}^{h_n+1} \left(2^{h_n+1-i} - \left\lceil \frac{2^{h_n}-p_n-2^{i-1}+1}{2^i} \right\rceil\right) \cdot f(i)$, if $0 < p_n < 2^{h_n}$, i.e., $2^{h_n} < n < 2^{h_n+1}$.
\end{enumerate}
\end{Cor}
\begin{proof}
By Theorem \ref{Theo:Min_H_gfb_characterization}, Part 1, the gfb-tree minimizes the HM on $\BT$. Furthermore, by Proposition \ref{Prop:gfb_number_subtrees_certain_height}, the gfb-tree with $n$ leaves has $2^{h_n-i}$ pending subtrees of height $i$ if $p_n = 0$, and $2^{h_n+1-i} - \left\lceil \frac{2^{h_n}-p_n-2^{i-1}+1}{2^i} \right\rceil$ pending subtrees of height $i$ if $0 < p_n < 2^{h_n}$. Summing over all possible heights $1 \leq i \leq h\left(T^{gfb}_n\right)$ yields the stated expressions. This completes the proof.
\end{proof}

It remains to determine the minimum value of the HM on $\T$, which is established in the following corollary.

\begin{Cor}
\label{Cor:Meta_Min_val_T}
Let $n \geq 2$ and let $f$ be a $1$-positive function. Then the minimum value of the height metaconcept $\Phi^{\mathcal{H}}_f$ on $\T$ is $\Phi^{\mathcal{H}}_f\left(T^{star}_n\right) = f(1)$.
\end{Cor}
\begin{proof}
By Proposition \ref{Prop:star_Min_on_T}, the star tree minimizes the HM on $\T$. Moreover, $\Phi^{\mathcal{H}}_f\left(T^{star}_n\right) = f(1)$. This completes the proof.
\end{proof}

In the following section, we analyze the locality and the recursive structure of the HM.

\subsubsection{Locality and recursiveness}

We show that the HM is not necessarily local but is recursive for all functions $f$. This generalizes an earlier result for the $B_1$ index (\citet[Propositions 10.2 and 10.3]{Fischer2023}) to the HM. We begin by analyzing the locality property.

\begin{Rem}
\label{Rem:H_not_local}
We construct a family of counterexamples demonstrating that the HM is not local.

Let $T = (T_v,\widehat{T})$ and $T' = (T'_v,\widehat{T})$ be two trees with the same number of leaves which differ only in one of their maximal pending subtrees (namely $T_v$ versus $T_v'$). Then
\begin{align*}
    \Phi^{\mathcal{H}}_f(T) - \Phi^{\mathcal{H}}_f(T') &= \sum\limits_{v \in \mathring{V}(T_v)} f\left(h_{T_v}(v)\right) + \sum\limits_{v \in \mathring{V}(\widehat{T})} f\left(h_{\widehat{T}}(v)\right) + f\left(h_T(\rho_T)\right)\\
    &- \left(\sum\limits_{v \in \mathring{V}(T'_v)} f\left(h_{T'_v}(v)\right) + \sum\limits_{v \in \mathring{V}(\widehat{T})} f\left(h_{\widehat{T}}(v)\right) + f\left(h_{T'}(\rho_{T'})\right)\right)\\
    &= \Phi^{\mathcal{H}}_f(T_v) - \Phi^{\mathcal{H}}_f(T'_v) + f\left(h(T)\right) - f\left(h(T')\right).
\end{align*}
Thus, the locality condition for $T$ and $T'$ holds if and only if their heights are equal. For a non-local  example, see Figure \ref{Fig:H_meta_not_local}.
\end{Rem}

\begin{figure}
    \centering
    \includegraphics[scale = 2]{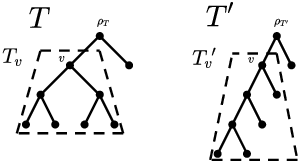}
    \caption{By Remark \ref{Rem:H_not_local}, this figure illustrates two trees $T$ and $T'$ that do not satisfy the locality condition for the height metaconcept. This is due to the fact that, while $T$ and $T'$ only differ in their maximal pending subtree $T_v$ versus $T_v'$, which both have four leaves, we get $h(T) \neq h(T')$. This shows that while the height sequences of $T_v$ and $T_v'$, namely $\mathcal{H}(T_v) = (1,1,2) $ and $\mathcal{H}(T_v') = (1,2,3)$ differ only in two positions, the height sequences of $T$ and $T'$ differ in three positions.}
    \label{Fig:H_meta_not_local}
\end{figure}

Second, we prove the recursiveness of the HM.

\begin{Prop}
\label{Prop:H_recursive}
The height metaconcept $\Phi^{\mathcal{H}}_f$ induces a recursive tree shape statistic for any function $f$.

Let $T \in \T$ be an arbitrary tree with standard decomposition $T = (T_1, \ldots, T_k)$, and let $n_i$ and $h_i$ denote the number of leaves and height, respectively, of the maximal pending subtree $T_i$. Let $f$ be an arbitrary function. 
\begin{itemize}
    \item If $n=1$, then $\Phi^{\mathcal{H}}_f(T) = 0$.
    \item If $n \geq 2$, then
    \[\Phi^{\mathcal{H}}_f(T) = \sum\limits_{i=1}^{k} \Phi^{\mathcal{H}}_{f}(T_i) + f\left(\max\left\{h_1, \ldots, h_k\right\} +1\right).\]
\end{itemize}
\end{Prop}

Notice that the proof of Proposition~\ref{Prop:H_recursive} proceeds analogously to \cite[Proof of Proposition 3.36]{Fischer2025}.

\begin{proof}
If $n = 1$, then the length of the height sequence is $0$. Consequently, the sum of the HM is taken over an empty set, and we have $\Phi^{\mathcal{H}}_f(T) = 0$.

For $n \geq 2$, the HM can be computed as follows:
\begin{align*}
    \Phi^{\mathcal{H}}_f(T) &= \sum\limits_{h \in \mathcal{H}(T)} f(h) \stackrel{\text{Obs. } \ref{Obs:H_recursiveness}}{=} \sum\limits_{h \in \mathcal{H}(T_1)} f(h) + \ldots + \sum\limits_{h \in \mathcal{H}(T_k)} f(h) + f(h(T))\\
    &= \sum\limits_{i=1}^{k} \Phi^{\mathcal{H}}_{f}(T_i) + f\left(\max\left\{h_1, \ldots, h_k\right\} +1\right).
\end{align*}
Thus, $\Phi^{\mathcal{H}}_f$ is a recursive tree shape statistic of length $x = 2$, with the recursions $r_1$ and $r_2$ given by:
\begin{itemize}
    \item HM: $\lambda_1 = \Phi^{\mathcal{H}}_{f}\left(T^{cat}_1\right) = 0$, and\\ $r_1\left(\left(r_1\left(T_1\right), r_2\left(T_1\right)\right), \ldots, \left(r_1\left(T_k\right), r_2\left(T_k\right)\right)\right) = \Phi^{\mathcal{H}}_{f}(T_1) + \ldots + \Phi^{\mathcal{H}}_{f}(T_k) + f\left(\max\left\{h_1, \ldots, h_k\right\} +1\right)$;
    \item tree height: $\lambda_2 = 0$, and $r_2\left(\left(r_1\left(T_1\right), r_2\left(T_1\right)\right), \ldots, \left(r_1\left(T_k\right), r_2\left(T_k\right)\right)\right) = \max\left\{h_1, \ldots, h_k\right\} +1$.
\end{itemize}
Hence, $\lambda \in \R^2$ and $r_i: \underbrace{\R^2 \times \ldots \times \R^2}_{k \text{ times}} \rightarrow \R$. Moreover, all recursions $r_i$ are symmetric. This completes the proof.
\end{proof}

We now turn our attention to the $B_1$ index.

\subsection{\texorpdfstring{$B_1$}{B1} index}
\label{Subsec:B1}

In this section, we finally resolve six open problems posed by \citet{Fischer2023}. Specifically, we characterize the trees that maximize the $B_1$ index on $\BT$ and $\T$, determine the number of such trees for various leaf numbers, and compute the maximum value of the $B_1$ index on both $\BT$ and $\T$. We begin by examining the number of trees that achieve this maximum. Note that the results from the last section cannot be applied here directly. This is because the $B_1$ index does not include the height value of the root and thus is induced by the third-order metaconcept (see Remark \ref{Rem:B1_HM}). The results for the HM nonetheless turn out to be useful for the analysis of the $B_1$ index. In fact, as we will subsequently point out, the following theorem is an important means to generalize our previous results to the $B_1$ index.

\begin{Theo}
\label{Theo:B1_extrema}
Let $B^{min}_{\mathcal{H}}$ denote the set of trees that minimize the height metaconcept $\Phi^{\mathcal{H}}_f$ on $\BT$ for all (strictly) increasing functions $f$, and let $B^{max}_1$ denote the set of trees that maximize the $B_1$ index on $\T$. Then $B^{min}_{\mathcal{H}} = B^{max}_1$. In particular, the trees that maximize the $B_1$ index on $\T$ are binary. Moreover, the gfb-tree is always among the maximizing trees, and it is the unique tree maximizing the $B_1$ index if and only if $n = 2^h - 2^i$ for some integers $h > i \geq 0$. Furthermore, if $n$ is odd, a unique maximizing tree exists if and only if $n = 2^h -1$ for some $h \geq 1$.
\end{Theo}

\begin{Rem}
By Theorem \ref{Theo:B1_extrema}, all results concerning the trees that minimize the HM on $\BT$ can be directly transferred to the trees that maximize the $B_1$ index on $\T$. In particular, all maximizing trees of the $B_1$ index are binary, even on $\T$, and the gfb-tree and the echelon tree (Proposition \ref{Prop:H_echelon_Min}) are among them for all $n \in \N$. 

Moreover, by Theorem \ref{Theo:Min_H_gfb_characterization}, Part 2(i), all maximizing trees share the same height sequence, and in particular, they all have minimal height $\left\lceil\log_2(n)\right\rceil$. Furthermore, if $T^{max}$ maximizes the $B_1$ index, then $\mathcal{H}(T^{max})_i \leq \mathcal{H}\left(T\right)_i$ for all $i = 1, \ldots, n-1$ and all binary trees $T \in \BT$.

By Lemma \ref{Lem:T_Min_H_properties}, all leaves (except for one if $n$ is odd) belong to a cherry in a tree that maximizes the $B_1$ index. Moreover, Propositions \ref{Prop:H_Min_n-1_to_Min_n} and \ref{Prop:H_Min_n_to_Min_n-1} show that all maximizing trees with an odd number of leaves $n-1$ can be constructed from maximizing trees with $n$ leaves, and vice versa. By Corollary \ref{Cor:Min_n_2n}, the same principle holds for leaf numbers $n$ and $2n$. 
\end{Rem}

\begin{proof}
Let $B^{min}_{\mathcal{H}}$ denote the set of trees that minimize the HM on $\BT$ for all (strictly) increasing functions $f$, and let $B^{max}_1$ denote the set of trees that maximize the $B_1$ index on $\T$.

We begin by proving that any tree $T^{max}$ that maximizes the $B_1$ index must be binary. Let $T \in \T \setminus \BT$, i.e., assume $T$ has a vertex $u_1$ with at least three children. Our goal is to construct a tree $T'$ from $T$ by adding one inner vertex in a specific way such that $B_1(T) < B_1(T')$. By iteratively applying this procedure, we eventually obtain a binary tree, which completes this part of the proof.

We use the same construction as in Lemma \ref{Lem:H_seq_caterpillar} and Figure \ref{Fig:H_caterpillar_T_to_binary}. Moreover, we denote the vertices on the path from $u_1$ to $\rho$ as $u_1, \ldots, u_m = \rho$ with $m \geq 1$, and assume, without loss of generality, that $h(T_{v_1}) \geq h(T_{v_i})$ for all $i = 2, \ldots, k$. Next, we distinguish three cases.

First, assume that $h(T_{v_1}) > h(T_{v_{2}})$. By construction, the vertices $u_1, \ldots, u_m = \rho$ retain the same height values in $T$ as in $T'$. Moreover, the vertex $u_0$ is newly added, so that, \[B_1(T') = B_1(T) + \frac{1}{h_{T'}(u_0)} > B_1(T).\]

Second, assume that $m = 1$, i.e., $u_1 = \rho$ (recall that the root does not contribute to the $B_1$ index). By construction, we again have
\[B_1(T') = B_1(T) + \frac{1}{h_{T'}(u_0)} > B_1(T).\]

Third, assume that $h(T_{v_1}) = h(T_{v_{2}})$ and $m \geq 2$. Then there exists an index $l \in \left\{1, \ldots, m\right\}$ such that the vertices $u_1, \ldots, u_l$ increase their height values by one, while the vertices $u_{l+1}, \ldots, u_m = \rho$ (if any) retain the same height values in $T'$ as in $T$. Consequently, for $i = 1, \ldots, l$, we have
\[h_T(u_i) = h_{T'}(u_{i}) -1 = h_{T'}(u_{i-1}).\] 
If $l = m \geq 2$, i.e., all vertices $u_1, \ldots, u_m = u_l = \rho$ increase their height values, then 
\[B_1(T') = B_1(T) + \frac{1}{h_{T'}(u_{l-1})}.\] 
If $l < m$, we have \[B_1(T') = B_1(T) + \frac{1}{h_{T'}(u_{l})}.\] 
In both cases, it follows that 
\[B_1(T') > B_1(T).\]

Thus, all trees that maximize the $B_1$ index are binary. Therefore, in the following, we restrict our attention to binary trees. Next, we show that the two sets, $B^{min}_{\mathcal{H}}$ and $B^{max}_1$ coincide.

Let $T^{min} \in B^{min}_{\mathcal{H}}$, i.e., let $T^{min}$ be a tree that minimizes the HM for all increasing functions $f$. We need to show that $T^{min}$ also maximizes the $B_1$ index, that is, $T^{min} \in B^{max}_1$. By Theorem \ref{Theo:Min_H_gfb_characterization}, Part 2(i), we have $\mathcal{H}(T^{min})_i \leq \mathcal{H}\left(T\right)_i$ for all $i = 1, \ldots, n-1$ and all binary trees $T \in \BT$. Therefore,
\[B_1(T^{min}) = \sum\limits_{i=1}^{n-2} \frac{1}{\mathcal{H}(T^{min})_i} \geq \sum\limits_{i=1}^{n-2} \frac{1}{\mathcal{H}(T)_i} = B_1\left(T\right).\]
It follows that $T^{min} \in B^{max}_1$, i.e., $T^{min}$ maximizes the $B_1$ index.

Let $T^{max} \in B^{max}_1$, i.e., let $T^{max}$ be a tree that maximizes the $B_1$ index. We need to show that $T^{max}$ also minimizes the HM for all increasing functions $f$, that is, $T^{max} \in B^{min}_{\mathcal{H}}$. By the previous part of the proof and Theorem \ref{Theo:Min_H_gfb_characterization}, Part 2(i), we know that $T^{gfb}_n \in B^{max}_1$, and thus
\[B_1\left(T^{max}\right) = \sum\limits_{i=1}^{n-2} \frac{1}{\mathcal{H}(T^{max})_i} = \sum\limits_{i=1}^{n-2} \frac{1}{\mathcal{H}\left(T^{gfb}_n\right)_i} = B_1\left(T^{gfb}_n\right).\]
By Theorem \ref{Theo:Min_H_gfb_characterization}, Part 1, we have
\[\mathcal{H}\left(T^{gfb}_n\right)_i \leq \mathcal{H}(T^{max})_i \text{ for all } i = 1, \ldots, n-1\]
and thus
\[\mathcal{H}\left(T^{gfb}_n\right)_i = \mathcal{H}(T^{max})_i \text{ for all } i = 1, \ldots, n-2.\]
It remains to show that $T^{max}$ has the same height as the gfb-tree, i.e., $\mathcal{H}\left(T^{gfb}_n\right)_{n-1} = \mathcal{H}(T^{max})_{n-1}$. 
Recall that the height of a tree equals the height of its tallest pending subtree plus one, i.e., it is the maximum of the first $n-2$ entries of the height sequence plus one.
Since $T^{gfb}_n$ and $T^{max}$ share the first $n-2$ values, they must also share the $(n-1)$-th value, which corresponds to the tree's height. Hence, $T^{max} \in B^{min}_{\mathcal{H}}$, i.e., $T^{max}$ minimizes the HM.

Finally, the statement for the unique maximization of the gfb-tree is a direct consequence of Proposition \ref{Prop:gfb_unique_Min}, Theorem \ref{Theo:Min_H_gfb_characterization} Part 2.ii), and the first part of the proof ($B^{min}_{\mathcal{H}} = B^{max}_1$). This completes the proof. 
\end{proof}

We now calculate the maximum value of the $B_1$ index on $\BT$ and $\T$. For $n$ a power of two, this value coincides with the tight bound given by \citet[Proposition 10.1]{Fischer2023} for $\BT$.

\begin{Cor}
\label{Cor:B1_Max_val}
Let $n \geq 2$ with $n = 2^{h_n} + p_n$, where $h_n = \lfloor \log_2(n) \rfloor$ and $0 \leq p_n < 2^{h_n}$. Then the maximum value of the $B_1$ index on $\BT$ and $\T$ is given by
\begin{enumerate}[(i)]
    \item $\sum\limits_{i = 1}^{h_n-1} 2^{h_n-i} \cdot \frac{1}{i}$, if $p_n = 0$ (i.e., $n = 2^{h_n}$), and
    \item $\sum\limits_{i = 1}^{h_n} \left(2^{h_n+1-i} - \left\lceil \frac{2^{h_n}-p_n-2^{i-1}+1}{2^i} \right\rceil\right) \cdot \frac{1}{i}$, if $0 < p_n < 2^{h_n}$ (i.e., $2^{h_n} < n < 2^{h_n+1}$).
\end{enumerate}
\end{Cor}
\begin{proof}
Observe that the root is the unique inner vertex with
\begin{enumerate}[(1)]
    \item height value $h_n$ if $p_n = 0$, and
    \item height value  $h_n +1$ if $p_n > 0$.
\end{enumerate}
The result then follows directly from Proposition \ref{Prop:gfb_number_subtrees_certain_height} and Theorem \ref{Theo:B1_extrema}.
\end{proof}

As noted in Remark \ref{Rem:B1_HM}, the $B_1$ index is induced by the third-order HM. In the next section, we introduce three new (im)balance indices, which are induced by the first-order and second-order HM, respectively.

\subsection{New imbalance indices}
\label{Subsec:new_imb_ind}

Based on the results of the HM, we observe that it induces an entire family of imbalance indices. In this section, we focus on three representatives of this family by introducing three new (im)balance indices. The first one can be thought of as a modification of the $B_1$ index which includes the root, whereas the other ones are based on functions that already occur in the literature in connection with other (im)balance indices and metaconcepts (cf. \citet{Fischer2025}).

\begin{Def} First, we introduce the \textit{balance index} obtained from the first-order HM with $\widehat{f}(h) = \frac{1}{h}$, which we call the $\widehat{B}_1$ index. Note that $\widehat{f} = \id^{-1}$ is also employed by the $B_1$ index, even if in a slightly different summation. Let $\widehat{B}_1: \T \rightarrow \R$ be the function, defined for $T \in \T$, by
\[\widehat{B}_1(T) = \sum\limits_{v \in \mathring{V}(T)} \frac{1}{h_T(v)} = \sum\limits_{h \in \mathcal{H}(T)} \frac{1}{h} = \Phi^{\mathcal{H}}_{\widehat{f}}(T).\]
Note that, unlike the $B_1$ index, the $\widehat{B}_1$ index includes the height value of the root, i.e., the height of the tree. In other words, \[\widehat{B}_1(T) = B_1(T) + \frac{1}{h(T)}.\]

Second, we define the \textit{imbalance index} obtained from the first-order HM with $\widetilde{f}(h) = h$, which we call the $\widetilde{B}_1$ index. Note that $\widetilde{f} = \id$ is also employed by the so-called Sackin \cite{Sackin1972,Fischer2021a} and Colless \cite{Colless1982,Coronado2020a} indices (\cite{Fischer2025}), albeit with different underlying sequences, namely the clade size and the balance value sequence, respectively. Let $\widetilde{B}_1: \T \rightarrow \R$ be the function defined, for $T \in \T$, by
\[\widetilde{B}_1(T) = \sum\limits_{v \in \mathring{V}(T)} h_T(v) = \sum\limits_{h \in \mathcal{H}(T)} h = \Phi^{\mathcal{H}}_{\widetilde{f}}(T).\]

Third, we introduce a normalized version of the $\widetilde{B}_1$ index, called the $\overline{B}_1$ index. This index is induced by the second-order HM with $\overline{f}(h,n) = \frac{1}{n} \cdot h$. Note that $\overline{f} = \frac{1}{n} \cdot \id$ is the function employed by the so-called average leaf depth \cite{Kirkpatrick1993,Shao1990} (using the clade size instead of the height value), which is also an imbalance index. Let $\overline{B}_1: \T \rightarrow \R$ be the function defined, for $T \in \T$, by
\[\overline{B}_1(T) = \frac{1}{n} \cdot \sum\limits_{v \in \mathring{V}(T)} h_T(v) = \frac{1}{n} \cdot \sum\limits_{h \in \mathcal{H}(T)} h = \Phi^{\mathcal{H}}_{\overline{f}}(T).\]
\end{Def}

We note that two tree shape statistics closely related to the $\widetilde{B}_1$ index have already appeared in the literature, albeit in a slightly different setting. In particular, in the context of $k$-ary trees (that is, rooted trees in which each inner vertex has at most $k$ children), \citet{Cha2012complete, Cha2012integer} introduced the following two quantities:
\[sumh(T) = \sum\limits_{v \in V(T)} h_v+1 \quad \text{ and } \quad sumh'(T) = \sum\limits_{v \in V(T)} h_v.\]
The statistic $\mathrm{sumh}'(T)$ coincides with $\widetilde{B}_1(T)$ when the latter is generalized to $k$-ary trees in the sense of~\cite{Cha2012complete,Cha2012integer}. The author derived several integer sequences corresponding to the values of $sumh$ and $sumh'$ for specific classes of $k$-ary trees  which the author considers as \enquote{balanced} (e.g., concerning height or clade size aspects). However, in \citet{Cha2012complete, Cha2012integer}, the author did not consider tree balance as a quantifiable measure. As discussed in Remark~\ref{Rem:oeis} below, one of the integer sequences identified in~\cite{Cha2012complete,Cha2012integer} also arises in our setting; we will explain this connection there.

\begin{Rem}
\label{Rem:new_indices}
    First, note that $\widehat{f}(h) = \frac{1}{h}$ is strictly decreasing. Consequently, a tree maximizes (resp. minimizes) the $\widehat{B}_1$ index if and only if it minimizes (resp. maximizes) the first-order HM with function $-\widehat{f}(h) = -\frac{1}{h}$. Hence, by Theorem \ref{Theo:H_imbalance_index}, $\widehat{B}_1$ is a balance index on $\BT$, but not on $\T$. The latter is true as -- although $-\widehat{f}(h)$ is strictly increasing -- it is not $1$-positive. Due to this relationship, all results concerning maximizing and minimizing trees of the HM with a strictly increasing function $f$ can be transferred to the $\widehat{B}_1$ index by switching the roles of minimization and maximization. Moreover, by Theorem \ref{Theo:Min_H_gfb_characterization}, $\widehat{B}_1$ is maximized by the gfb-tree, and all other maximizing trees share the same height sequence.

    Second, by definition, the $\widetilde{B}_1$ index is induced by the first-order HM with the strictly increasing and $1$-positive function $\widetilde{f}(h) = h$, while the $\overline{B}_1$ index is induced by the second-order HM. Moreover, the two indices $\widetilde{B}_1$ and $\overline{B}_1$ are equivalent (note, however, that this is only true if any pair of trees under investigation has the same number of leaves), since \[\overline{B}_1(T) = \frac{1}{n} \cdot \widetilde{B}_1(T)\] (for more details, see Remark \ref{Rem:H_foMeta_equiv_higherorderMeta}). Hence, by Theorem \ref{Theo:H_imbalance_index}, both $\widetilde{B}_1$ and $\overline{B}_1$ are imbalance indices on $\BT$ and on $\T$. Furthermore, due to their relationship with the HM, all results concerning maximizing and minimizing trees for strictly increasing and $1$-positive $f$ carry over directly to these indices. The most notable property is that the gfb-tree minimizes these indices on $\BT$, and all other minimizing trees on $\BT$ share the same height sequence (Theorem \ref{Theo:Min_H_gfb_characterization}).
\end{Rem}

\begin{Rem}
    As seen above, the indices $B_1$, $\widehat{B}_1$ and $\widetilde{B}_1$ are all induced by the HM and share the same extremal trees on $\BT$. However, they are not equivalent, as they can rank trees within $\BT$ differently.

    For example, consider the trees $T_1$ and $T_2$ in Figure \ref{Fig:B1_indices_rankings_T1_T2} for a comparison of the $B_1$ and $\widehat{B}_1$ indices (both of which are balance indices). Their height sequences are $\mathcal{H}(T_1) = (1,1,1,1,2,3,4,5,6,7)$ and $\mathcal{H}(T_2) = (1,1,1,2,2,2,3,3,4,5)$. Hence, we have
    \[B_1(T_1) = 4\cdot1 + \frac{1}{2} + \frac{1}{3} + \frac{1}{4} + \frac{1}{5} + \frac{1}{6} = 5.45 > 5.41\overline{6} = 3\cdot1 + 3\cdot\frac{1}{2} + 2\cdot\frac{1}{3} + \frac{1}{4} = B_1(T_2),\]
    but
    \[\widehat{B}_1(T_1) = B_1(T_1) + \frac{1}{7} \approx 5.59 < 5.62 \approx B_1(T_2) + \frac{1}{5} = \widehat{B}_1(T_2).\]
    Thus, $T_1$ is more balanced than $T_2$ according to the $B_1$ index, but less balanced according to the $\widehat{B}_1$ index.

    On the other hand, $\widetilde{B}_1$ ranks the trees $T_3$ and $T_4$ in Figure \ref{Fig:B1_indices_rankings_T3_T4} differently than $B_1$ and $\widehat{B}_1$. Their height sequences are $\mathcal{H}(T_3) = (1,1,1,2,3,4,5)$ and $\mathcal{H}(T_4) = (1,1,2,2,3,3,4)$. Hence, we have
    \[B_1(T_3) = 3\cdot1 + \frac{1}{2} + \frac{1}{3} + \frac{1}{4} = \frac{49}{12} > \frac{44}{12} = 2\cdot1 + 2\cdot\frac{1}{2} + 2\cdot\frac{1}{3} = B_1(T_4),\]
    as well as
    \[\widehat{B}_1(T_3) = B_1(T_3) + \frac{1}{5} = \frac{257}{60} > \frac{235}{60} = B_1(T_4) + \frac{1}{4} = \widehat{B}_1(T_4),\]
    but
    \[\widetilde{B}_1(T_3) = 3\cdot1+2+3+4+5 = 17 > 16 = 2\cdot1 + 2\cdot2 + 2 \cdot 3 + 4 = \widetilde{B}_1(T_4).\]
    Note that $B_1$ and $\widehat{B}_1$ are balance indices, whereas $\widetilde{B}_1$ is an imbalance index. Consequently, $T_3$ is more balanced than $T_4$ according to $B_1$ and $\widehat{B}_1$, but less balanced according to $\widetilde{B}_1$.
\end{Rem}

\begin{figure}
    \centering
    \includegraphics[scale=1.3]{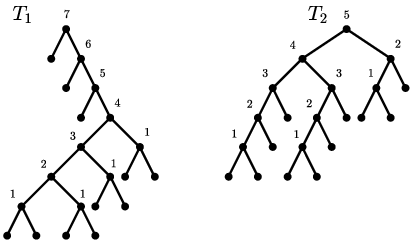}
    \caption{Two trees with $n = 11$ leaves that are ranked differently by the $B_1$ and $\widehat{B}_1$ indices. The inner vertices are labeled with their height value. We have $\mathcal{H}(T_1) = (1,1,1,1,2,3,4,5,6,7)$ and $\mathcal{H}(T_2) = (1,1,1,2,2,2,3,3,4,5)$.}
    \label{Fig:B1_indices_rankings_T1_T2}
\end{figure}
\begin{figure}
    \centering
    \includegraphics[scale=1.5]{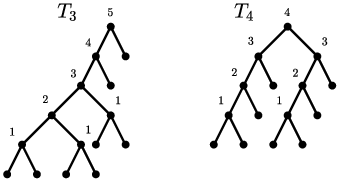}
    \caption{Two trees with $n = 8$ leaves that are ranked differently by the $B_1$ and $\widehat{B}_1$ indices compared to the $\widetilde{B}_1$ index. The inner vertices are labeled with their height value. We have $\mathcal{H}(T_3) = (1,1,1,2,3,4,5)$ and $\mathcal{H}(T_4) = (1,1,2,2,3,3,4)$.}
    \label{Fig:B1_indices_rankings_T3_T4}
\end{figure}

\paragraph{Recursiveness}
$\widehat{B}_1$, $\widetilde{B}_1$, and  $\overline{B}_1$ are recursive. Let $T \in \T$ be an arbitrary tree with standard decomposition $T = (T_1, \ldots, T_k)$, and let $n_i$ and $h_i$ denote the leaf number and height, respectively, of the maximal pending subtree $T_i$. Then, for $n=1$, we have \[\widehat{B}_1(T) = \widetilde{B}_1(T) = \overline{B}_1(T) = 0,\] and for $n \geq 2$, the indices can be calculated recursively as follows. Note that the recursion for $\widehat{B}_1$ holds for arbitrary trees, i.e., $k > 2$, but recall that $\widehat{B}_1$ is a balance index only in the binary case, i.e., if $k = 2$.
\begin{align*}
    &\widehat{B}_1(T) \stackrel{\text{Prop. \ref{Prop:H_recursive}}}{=} \sum\limits_{i=1}^{k} \widehat{B}_1(T_i) + \frac{1}{\max\left\{h_1, \ldots, h_k\right\} +1};\\
    &\widetilde{B}_1(T) \stackrel{\text{Prop. \ref{Prop:H_recursive}}}{=} \sum\limits_{i=1}^{k} \widetilde{B}_1(T_i) + \max\left\{h_1, \ldots, h_k\right\} +1;\\
    &\overline{B}_1(T) = \frac{1}{n} \cdot \left(\sum\limits_{i=1}^{k} n_i \cdot \overline{B}_1(T_i) + \max\left\{h_1, \ldots, h_k\right\} +1\right).
\end{align*}

\paragraph{Extremal values}
We can calculate the extremal values of all new (im)balance indices. Therefore, recall Remark \ref{Rem:new_indices}, i.e., their relationship to the HM. The values are listed in Tables \ref{Tab:new_indices_extr_val_1} and \ref{Tab:new_indices_extr_val_2}.

\begin{table}[htbp]
    \caption{Extremal values of the $\widehat{B}_1$, $\widetilde{B}_1$, and $\overline{B}_1$ indices where $n \geq 2$. Note that the values on the left follow from Corollary \ref{Cor:Meta_Max_val}, and the values on the right follow from Corollary \ref{Cor:Meta_Min_val_T}.}
    \label{Tab:new_indices_extr_val_1}
    \centering
    \renewcommand*{\arraystretch}{2.5}
    \begin{tabular}{c|c||c|c}
        $\min\limits_{T^b \in \BT} \widehat{B}_1\left(T^b\right)$ & $\sum\limits_{i = 1}^{n-1} \frac{1}{i}$ & $\min\limits_{T \in \T} \widehat{B}_1(T)$ & $1$\\
        \hline
        $\max\limits_{T \in \T} \widetilde{B}_1(T) = \max\limits_{T^b \in \BT} \widetilde{B}_1\left(T^b\right)$ & $\sum\limits_{i = 1}^{n-1} i = \frac{n(n-1)}{2}$ & $\min\limits_{T \in \T} \widetilde{B}_1(T)$ & $1$\\
        \hline
        $\max\limits_{T \in \T} \overline{B}_1(T) = \max\limits_{T^b \in \BT} \overline{B}_1\left(T^b\right)$ & $\frac{1}{n} \cdot \frac{n(n-1)}{2} = \frac{n-1}{2}$ & $\min\limits_{T \in \T} \overline{B}_1(T)$ & $\frac{1}{n}$
    \end{tabular}
\end{table}

\begin{table}[htbp]
    \caption{Extremal values of the $\widehat{B}_1$, $\widetilde{B}_1$, and $\overline{B}_1$ indices for  $n \geq 2$ with $n = 2^{h_n}+p_n$, where $h_n = \lceil \log_2(n) \rceil$ and $0 \leq p_n < 2^{h_n}$. These values follow directly from Corollary \ref{Cor:Meta_Min_val}. Moreover, in a manner analogous to the proof of Theorem \ref{Theo:B1_extrema} for the $B_1$ index, it can be shown that the maximum value of the $\widehat{B}_1$ index is the same when computed over the set of rooted binary trees $\BT$ and the set of all rooted trees $\T$.}
    \label{Tab:new_indices_extr_val_2}
    \centering
    \renewcommand*{\arraystretch}{2.5}
    \begin{tabular}{c||c|c}
        & $p_n = 0$ (i.e., $n = 2^{h_n}$) & $0 < p_n < 2^{h_n}$ (i.e., $2^{h_n} < n < 2^{h_n+1}$)\\
        \hline \hline
        $\max\limits_{T \in \BT} \widehat{B}_1(T) = \max\limits_{T \in \T} \widehat{B}_1(T)$ & $\sum\limits_{i = 1}^{h_n} 2^{h_n-i} \cdot \frac{1}{i}$
        & $\sum\limits_{i = 1}^{h_n+1} \left(2^{h_n+1-i} - \left\lceil \frac{2^{h_n}-p_n-2^{i-1}+1}{2^i} \right\rceil\right) \cdot \frac{1}{i}$\\
        \hline
        $\min\limits_{T \in \BT} \widetilde{B}_1(T)$ & $\sum\limits_{i = 1}^{h_n} 2^{h_n-i} \cdot i$ & $ \sum\limits_{i = 1}^{h_n+1} \left(2^{h_n+1-i} - \left\lceil \frac{2^{h_n}-p_n-2^{i-1}+1}{2^i} \right\rceil\right) \cdot i$\\
        \hline
        $\min\limits_{T \in \BT} \overline{B}_1(T)$ & $\frac{1}{n} \cdot \sum\limits_{i = 1}^{h_n} 2^{h_n-i} \cdot i$ & $\frac{1}{n} \cdot \sum\limits_{i = 1}^{h_n+1} \left(2^{h_n+1-i} - \left\lceil \frac{2^{h_n}-p_n-2^{i-1}+1}{2^i} \right\rceil\right) \cdot i$\\
    \end{tabular}
\end{table}

Having established explicit formulas for the maximum and minimum values of the three indices, we can now present a recursive formula for the minimum values of the $\widetilde{B}_1$ and $\overline{B}_1$ indices on $\BT$. In particular, we compute the $\widetilde{B}_1$ value of the gfb-tree recursively. This result allows us to connect the minimum values of the $\widetilde{B}_1$ index to \citet[Sequence A005187]{OEIS}. Note that the extremal values of the other indices cannot be linked to integer sequences, since they may take non-integer values.

\begin{Prop}
\label{Prop:gfb_widetildeB1_recursion}
    Let $n \geq 1$. Then $\widetilde{B}_1\left(T^{gfb}_1\right) = 0$, and for $n \geq 2$, we have \\ $\widetilde{B}_1\left(T^{gfb}_n\right) = \widetilde{B}_1\left(T^{gfb}_{\left\lfloor\frac{n+1}{2}\right\rfloor}\right) + n-1.$
    Moreover, we have $\overline{B}_1\left(T^{gfb}_n\right)= \frac{1}{n}\widetilde{B}_1\left(T^{gfb}_n\right)$.
\end{Prop}
\begin{proof}
    We start by proving the recursion for the $\widetilde{B}_1$ index.

    If $n = 1$, we have $\widetilde{B}_1\left(T^{gfb}_1\right) = 0$ as $T^{gfb}_1$ only contains of a single leaf, which turns $\widetilde{B}_1\left(T^{gfb}_1\right)$ into an empty sum. This establishes the base case of the recursion.

    Now we distinguish between the cases where $n$ is even or odd.

    \begin{enumerate}
        \item We first consider the case that $n$ is even. We need to show that
        \[\widetilde{B}_1\left(T^{gfb}_n\right) = \widetilde{B}_1\left(T^{gfb}_{\left\lfloor\frac{n+1}{2}\right\rfloor}\right) + n-1 = \widetilde{B}_1\left(T^{gfb}_{\frac{n}{2}}\right) + n-1.\]
        By Remark \ref{Rem:gfb_properties}, the gfb-tree with $n$ leaves can be obtained from the gfb-tree with $\frac{n}{2}$ leaves by attaching a cherry to each of its $\frac{n}{2}$ leaves. This increases each of the original $\frac{n}{2}-1$ height values by one. Moreover, each parent of the attached cherries has height value $1$. Hence we obtain
        \[\widetilde{B}_1\left(T^{gfb}_n\right) = \left(\widetilde{B}_1\left(T^{gfb}_{\frac{n}{2}}\right) + \frac{n}{2}-1\right) + \frac{n}{2} = \widetilde{B}_1\left(T^{gfb}_{\frac{n}{2}}\right) + n-1.\]
        This completes the proof for even $n$.

        \item Next, suppose that $n$ is odd. We  need to show that
        \[\widetilde{B}_1\left(T^{gfb}_n\right) = \widetilde{B}_1\left(T^{gfb}_{\left\lfloor\frac{n+1}{2}\right\rfloor}\right) + n-1 = \widetilde{B}_1\left(T^{gfb}_{\frac{n+1}{2}}\right) + n-1.\]
        By Remark \ref{Rem:gfb_properties}, the gfb-tree with $n+1$ leaves can be obtained from the gfb-tree with $n$ leaves by attaching a cherry to the unique leaf that is not part of a cherry. Hence, \[\widetilde{B}_1\left(T^{gfb}_n\right) = \widetilde{B}_1\left(T^{gfb}_{n+1}\right) -1.\]
        Applying the same argument as in the even case to $n+1$, we have \[\widetilde{B}_1\left(T^{gfb}_{n+1}\right) = \left(\widetilde{B}_1\left(T^{gfb}_{\frac{n+1}{2}}\right) + \frac{n+1}{2} - 1\right) + \frac{n+1}{2} = \widetilde{B}_1\left(T^{gfb}_{\frac{n+1}{2}}\right) + n.\]
        Thus,
        \[\widetilde{B}_1\left(T^{gfb}_n\right) = \widetilde{B}_1\left(T^{gfb}_{n+1}\right) -1 = \left(\widetilde{B}_1\left(T^{gfb}_{\frac{n+1}{2}}\right) + n\right) - 1 = \widetilde{B}_1\left(T^{gfb}_{\frac{n+1}{2}}\right) + n - 1.\]
        This completes the proof for odd $n$ and thus for the $\widetilde{B}_1$ index.
    \end{enumerate}

    The proof for the $\overline{B}_1$ index follows directly from Remark \ref{Rem:new_indices}. This completes the proof.
\end{proof}

This proposition immediately yields a recursive formula to calculate the minimum value of the $\widetilde{B}_1$ index on $\BT$.

\begin{Cor}
\label{Cor:widetildeB1_recursion}
    The minimum value of the $\widetilde{B}_1$ index on $\BT$ can be computed recursively. Let $\widetilde{b}_1(n) \coloneqq \min\limits_{T \in \BT}\widetilde{B}_1(T)$. Then $\widetilde{b}_1(1) = 0$, and for $n \geq 2$,
    \[\widetilde{b}_1(n) = \widetilde{b}_1\left(\left\lfloor\frac{n+1}{2}\right\rfloor\right) + n-1.\]
    Moreover, we have $\min\limits_{T \in \BT}\overline{B}_1(T) = \frac{1}{n}\widetilde{b}_1(n)$. 
\end{Cor}
\begin{proof}
    By Remark \ref{Rem:new_indices}, we have $\widetilde{b}_1(n) = \widetilde{B}_1\left(T^{gfb}_n\right)$, i.e., the gfb-tree always minimizes the $\widetilde{B}_1$ index on $\BT$. Now, the proof directly follows from Proposition \ref{Prop:gfb_widetildeB1_recursion}.
\end{proof}

We can use this result to link the minimum values of the $\widetilde{B}_1$ index to a sequence found in \citet{OEIS}.

\begin{Rem}
\label{Rem:oeis}
    The sequence of minimum values of the $\widetilde{B}_1$ index on $\BT$, namely, $0, 1, 3, 4, 7, 8, 10, 11, 15, 16, \ldots$, can be linked to \citet[Sequence A005187]{OEIS}. By Corollary \ref{Cor:widetildeB1_recursion}, we have $\widetilde{b}_1(n) = a(n-1)$ for all $n \geq 1$, where $\widetilde{b}_1(n) = \min\limits_{T \in \BT}\widetilde{B}_1(T)$ and $a(n)$ denotes the $n$-th entry of \citet[Sequence A005187]{OEIS}. We further remark that the sequence $a(n)$ has previously been related to sums of vertex height values in trees; see, for instance, \cite{Cha2012complete,Cha2012integer}. More precisely, for so-called \emph{complete $2$-ary trees} $T$ with $n$ vertices, the sum of vertex height values satisfies $sumh(T) = a(n)$. Note that, in our setting of rooted binary trees, complete trees correspond precisely to gfb-trees. The subtle difference between this earlier result and ours is that we restrict attention to trees in which every inner vertex has exactly two children, whereas \cite{Cha2012complete, Cha2012integer} allow each inner vertex to have at most two children. The reason for the shift of the sequences is thus twofold. First, the inner vertices of our trees with $n$ leaves correspond to the $2$-ary trees with $n-1$ vertices considered in \cite{Cha2012complete,Cha2012integer}. Second, the quantity $sumh$ sums height values shifted by one, that is, each vertex contributes its height value plus one. In particular, leaves of the tree contribute a value of $1$ rather than $0$. However, for such a complete $2$-ary tree $T$, the sequence $(sumh'(T))_{n}$ corresponds to \citet[Sequence A011371]{OEIS}, which repeats every term of Sequence A005187 twice, i.e., $(sumh'(T))_{n} = (0, 0, 1, 1, 3, 3, 4, 4, 7, 7, \ldots)$ for $n \geq 0$. In contrast to $sumh$, the quantity $sumh'$ sums the unshifted height values. Hence, leaves contribute a value of $0$, and the repetition in the sequence originates from the fact that adding a second leaf child to an inner vertex that already has one leaf child (thereby forming a cherry) does not change the height value of that vertex.
\end{Rem}

\section{Discussion}
\label{Sec:Discussion}

In this manuscript, we extended the set of metaconcepts for quantifying the balance of rooted trees, covering both binary and arbitrary rooted trees. Following \citet{Fischer2025}, we defined a tree shape sequence, the height sequence, and introduced a new class of metaconcepts, which we called the height metaconcept (HM). We showed that an imbalance index can be derived from each strictly increasing and $1$-positive function $f$ (Theorem \ref{Theo:H_imbalance_index}). Notably, these conditions are not very restrictive, as many existing imbalance indices employ such functions, albeit in the context of other input sequences. To name just a few examples: the well-known Sackin index \cite{Sackin1972,Fischer2021a,Fischer2025,Fischer2023} and the so-called $Q$-shape statistic \cite{Cleary2025,Fill1996} use $f = id$ and $f = \log$, respectively, but apply these functions to the clade size sequence rather than the height sequence. In Section \ref{Subsec:new_imb_ind} we have already introduced the imbalance index $\widetilde{B_1}$, which can be regarded as a Sackin-type height-based index, because it also uses the identity function. In a similar fashion, a $Q$-shape-type height-based index based on the logarithm can easily be derived. This shows our results on metaconcepts open the door to various new (im)balance indices with different properties. Investigating and comparing these further is an interesting direction for future research.

Furthermore, note that our studies show that it is sufficient for $f$ to be strictly increasing to obtain a \textit{binary} imbalance index from the HM (Theorem \ref{Theo:H_imbalance_index}). A similar result was established for the balance value metaconcept in \citet{Fischer2025}. In contrast, the two other metaconcept classes introduced by these authors yield valid imbalance indices only for a narrower set of functions $f$. We also showed that the set of minimizing trees of the HM differs from the respective sets of the earlier metaconcepts (compare Proposition \ref{Prop:H_echelon_Min} with the corresponding statements in \citet{Fischer2025}), hence enabling new imbalance indices with new and interesting properties.

We completely characterized the trees minimizing the HM for all increasing functions, as they all have the height sequence of the gfb-tree (Theorem \ref{Theo:Min_H_gfb_characterization}). Moreover, Proposition \ref{Prop:gfb_unique_Min} provides an exhaustive list of all leaf numbers for which the gfb-tree is the unique minimizer.

Furthermore, the $B_1$ index is induced by the third-order HM, as it does not take the root into account. Thus, deriving $B_1$ directly from the HM requires, in addition to the height values, knowledge of the number of inner vertices or leaves and the overall height of the tree. For this reason, most results from the analysis of the HM cannot be transferred to the $B_1$ index directly (Theorem \ref{Theo:B1_extrema}). Nevertheless, leveraging the close correspondence with the HM, we were able to resolve six open problems regarding the $B_1$ index posed by \citet{Fischer2023}. The first two problems concerned identifying the trees that maximize the $B_1$ index among arbitrary trees and among binary trees. We fully characterized these maximizing trees, showing that the gfb-tree always attains the maximum $B_1$ value (both on $\BT$ and $\T$), and that all other maximizing trees share the same height sequence as the gfb-tree. Using this characterization, we computed the maximum value of the $B_1$ index on $\BT$ and $\T$, and determined all leaf numbers for which a unique maximizing tree exists. In addition to the $B_1$ index, we introduced three new (im)balance indices induced by the HM and analyzed them exploiting their relationship to the HM.

As illustrated by the comparison of tree shape sequences, the underlying metaconcepts yield different outcomes depending on their interpretation of imbalance. Future research could further investigate the four metaconcept classes with respect to additional properties, such as resolution, i.e., the ability to assign a wide range of values to a set of trees in order to distinguish them. Notably, binary trees with identical height or leaf depth sequences can already occur with as few as five or six leaves, respectively, whereas at least nine leaves are required for two trees to share the same clade size sequence or balance value sequence (\citet{Fischer2025}). Even more strikingly, among the 23 binary trees with $n = 8$ leaves, only $4$ have a unique height sequence, while the remaining $19$ trees are distributed among just $4$ height sequences.

\section*{Acknowledgments}
The authors wish to thank Sophie Kersting for various discussions and helpful insights. They also thank the handling editor and reviewer for their helpful comments. Parts of this material are based upon work supported by the National Science Foundation under Grant No. DMS-1929284 while MF and KW were in residence at the Institute for Computational and Experimental Research in Mathematics in Providence, RI, during the Theory, Methods, and Applications of Quantitative Phylogenomics semester program.

\section*{Declarations}

\subsection*{Declaration of interests}
The authors declare that they have no known competing financial interests or personal relationships that could have appeared to influence the work reported in this paper.

\bibliographystyle{plainnat}
\bibliography{References}

\appendix
\section{Additional figures}
All inner vertices $v$ of the trees shown in Figures \ref{Fig:5_2_3}-\ref{Fig:11_194_199} are labeled either by their height value $h_v$, or by a tuple $(h_v,n_v,b_v)$ specifying their height value $h_v$, clade size $n_v$, and balance value $b_v$. When leaf labels are present, they indicate the corresponding leaf depths. As a remark, similar figures appear in \cite{Fischer2025}, but here we explicitly include the height sequences of the trees.

\begin{figure}[htbp]
    \centering
    \includegraphics[scale=2]{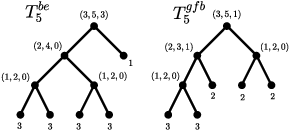}
    \caption{Unique minimal example of two binary trees with the same $\mathcal{H}$ but different $\mathcal{N}$, $\Delta$, and $\mathcal{B}$, respectively. Specifically, $n = 5$, $\mathcal{H}\left(T^{be}_5\right) = \mathcal{H}\left(T^{gfb}_5\right) = (1,1,2,3)$, $\mathcal{N}\left(T^{be}_5\right) = (2,2,4,5) \neq (2,2,3,5) = \mathcal{N}\left(T^{gfb}_5\right)$, $\Delta\left(T^{be}_5\right) = (1,3,3,3,3) \neq (2,2,2,3,3) = \Delta\left(T^{gfb}_5\right)$, and $\mathcal{B}\left(T^{be}_5\right) = (0,0,0,3) \neq (0,0,1,1) = \mathcal{B}\left(T^{gfb}_5\right)$.}
    \label{Fig:5_2_3}
\end{figure}

\begin{figure}[htbp]
    \centering
    \includegraphics[scale=2]{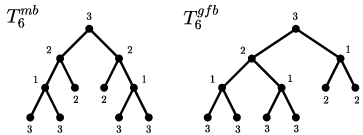}
    \caption{Unique minimal example of two binary trees with different $\mathcal{H}$ but the same $\Delta$. Here, $n = 6$, $\mathcal{H}\left(T^{mb}_6\right) = (1,1,2,2,3) \neq (1,1,1,2,3) = \mathcal{H}\left(T^{gfb}_6\right)$, and $\Delta\left(T^{mb}_6\right) = \Delta\left(T^{gfb}_6\right) = (2,2,3,3,3,3)$.}
    \label{Fig:6_5_6}
\end{figure}

\begin{figure}[htbp]
    \centering
    \includegraphics[scale=2]{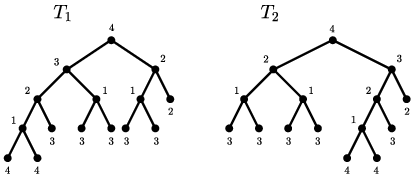}
    \caption{One of two minimal examples of two binary trees that share the same $\mathcal{H}$ and the same $\Delta$ simultaneously. Specifically, $n = 8$, $\mathcal{H}(T_1) = \mathcal{H}(T_2) = (1,1,1,2,2,3,4)$, and $\Delta(T_1) = (2,3,3,3,3,3,4,4) = \Delta(T_2)$.}
    \label{Fig:8_20_22}
\end{figure}

\begin{figure}[htbp]
    \centering
    \includegraphics[width=\textwidth]{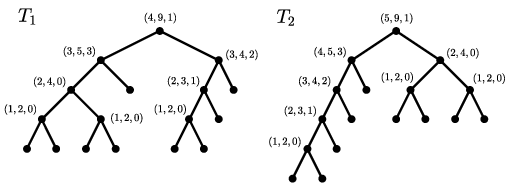}
    \caption{Unique minimal example of two binary trees having (1) different $\mathcal{H}$ but identical $\mathcal{N}$ and (2) different $\mathcal{H}$ but identical $\mathcal{B}$. Here, $n = 9$, $\mathcal{H}(T_1) = (1,1,1,2,2,3,3,4) \neq (1,1,1,2,2,3,4,5) = \mathcal{H}(T_2)$, $\mathcal{N}(T_1) = \mathcal{N}(T_2) = (2,2,2,3,4,4,5,9)$, and $\mathcal{B}(T_1) = \mathcal{B}(T_2) = (0,0,0,0,1,1,2,3)$.}
    \label{Fig:9_42_44}
\end{figure}

\begin{figure}[htbp]
    \centering
    \includegraphics[width=\textwidth]{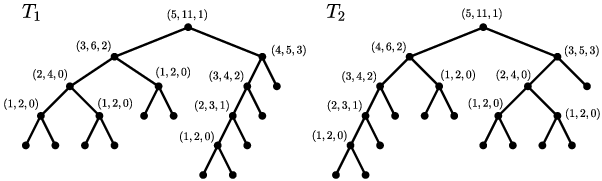}
    \caption{One of two minimal examples for two binary trees with the same $\mathcal{H}$, the same $\mathcal{N}$, and the same $\mathcal{B}$. Specifically, $n = 11$, $\mathcal{H}(T_1) = \mathcal{H}(T_2) = (1,1,1,1,2,2,3,3,4,5)$, $\mathcal{N}(T_1) = \mathcal{N}(T_2) = (2,2,2,2,3,4,4,5,6,11)$, and $\mathcal{B}(T_1) = \mathcal{B}(T_2) = (0,0,0,0,0,1,1,2,2,3)$.}
    \label{Fig:11_194_199}
\end{figure}

\end{document}